\documentclass[12pt]{amsart}
\usepackage{amsfonts,amsthm,amssymb}
\allowdisplaybreaks
\usepackage[mathscr]{eucal}
\usepackage{stmaryrd}
\usepackage{times}

\usepackage[curve,matrix,arrow]{xy} 
\xyoption{poly}

\usepackage[bookmarks=false, hyperfigures=false, a4paper]{hyperref}

\advance\textheight\topskip
\allowdisplaybreaks

\newenvironment{otherqg}[1]{
  \small #1}{
}

\newcommand{\mypar}[1]{\medskip\noindent\underline{#1}. \ }

\newcommand{\counit}{\epsilon}

\newcommand{\AT}{e^{-i\pi \frac{c}{12}}}

\newcommand{\pplus}{p}
\newcommand{\pminus}{p'}

\renewcommand{\tilde}{\widetilde}

\newcommand{\q}{\mathfrak{q}}

\newcommand{\idem}{\boldsymbol{e}}
\newcommand{\Idem}{\boldsymbol{e}}

\newcommand{\vSE}{\boldsymbol{v}^{\scriptscriptstyle\searrow}}
\newcommand{\vNW}{\boldsymbol{v}^{\scriptscriptstyle\nwarrow}}
\newcommand{\vSW}{\boldsymbol{v}^{\scriptscriptstyle\swarrow}}
\newcommand{\vNE}{\boldsymbol{v}^{\scriptscriptstyle\nearrow}}

\newcommand{\NWSE}{{\scriptscriptstyle\boldsymbol{\nwarrow\!\!\kern-1pt\searrow}}}

\newcommand{\wUp}{\boldsymbol{w}^{\uparrow}}
\newcommand{\wLeft}{\boldsymbol{w}^{\scriptstyle\leftarrow}}
\newcommand{\wRight}{\boldsymbol{w}^{\scriptstyle\rightarrow}}
\newcommand{\wDown}{\boldsymbol{w}^{\downarrow}}

\newcommand{\piUp}{\boldsymbol{\pi}^{\uparrow}\!\!}
\newcommand{\piLeft}{\boldsymbol{\pi}^{\scriptstyle\leftarrow}}
\newcommand{\piRight}{\boldsymbol{\pi}^{\scriptstyle\rightarrow}}
\newcommand{\piDown}{\boldsymbol{\pi}^{\downarrow}}
\newcommand{\nilp}{\boldsymbol{w}}

\newcommand{\tXX}{\widetilde{\repX}}

\renewcommand{\hat}{\widehat}

\newcommand{\bref}[1]{\textbf{\textup{\ref{#1}}}}

\newcommand{\hSL}[1]{\widehat{s\ell}(#1)}

\newcommand{\acts}{{\rightharpoondown}}

\newcommand{\cZ}{\mathsf{Z}}

\newcommand{\id}{\mathrm{id}}

\renewcommand{\geq}{\,{\geqslant}\,}
\renewcommand{\leq}{\,{\leqslant}\,}

\newcommand{\EXT}{\mathrm{Ext}^\bullet}
\newcommand{\Ext}{\mathrm{Ext}_{\rule{0pt}{9.5pt}\UresSL2}^1}

\newcommand{\coup}[2]{\langle#1,#2\rangle} 

\newcommand{\tensor}{\otimes}

\newcommand{\vectv}[1]{|#1\rangle}
\newcommand{\lc}{\rep{C}}

\newcommand{\algA}{\mathfrak{A}}
\newcommand{\algB}{\mathfrak{B}}
\newcommand{\algI}{\mathfrak{I}}

\newcommand{\UresSL}[1]{\overline{\mathscr{U}}_{\q} s\ell(#1)}
\newcommand{\XXX}{\mathscr{U}_{\pplus,\pminus}}

\newcommand{\ffrac}[2]{\mbox{\footnotesize$\displaystyle\frac{#1}{#2}$}}
\newcommand{\half}{%
  \mathchoice{\ffrac{1}{2}}{\frac{1}{2}}{\frac{1}{2}}{\frac{1}{2}}}

\newcommand{\qbin}[2]{\mathchoice%
  {{\qbinm{#1}{#2}}}{\qbinmm{#1}{#2}}%
  {\qbinmm{#1}{#2}}{\qbinmm{#1}{#2}}}
\newcommand{\qbinm}[2]{\mbox{\footnotesize$\displaystyle
    \genfrac{[}{]}{0pt}{}{#1}{#2}$}}
\newcommand{\qbinmm}[2]{\genfrac{[}{]}{0pt}{}{#1}{#2}}

\newcommand{\Radford}{\smash[t]{\widehat{\pmb{\phi}}}}
\newcommand{\Drinfeld}{\boldsymbol{\chi}}

\newcommand{\Rmin}{R_{\mathrm{min}}}
\newcommand{\Rproj}{R_{\mathrm{proj}}}

\newcommand{\RLeft}{R_{\boxbslash}}

\newcommand{\RRight}{R_{\boxslash}}

\newcommand{\setI}{\mathcal{I}}

\newcommand{\setR}{\mathcal{I}_1}
\newcommand{\setii}{\mathscr{I}_1}

\newcommand{\vvarphi}{\hat{\pmb{\varphi}}{}}

\newcommand{\modS}{\mathscr{S}}
\newcommand{\modT}{\mathscr{T}}

\newcommand{\projP}{\mathscr{P}}
\newcommand{\modL}{\mathscr{P}^{+}}

\newcommand{\repX}{\rep{X}}

\newcommand{\rep}{\mathscr}

\newcommand{\SLiiZ}{SL(2,\oZ)}
\newcommand{\oC}{\mathbb{C}}
\newcommand{\oP}{\mathbb{P}}

\newcommand{\oZ}{\mathbb{Z}}

\newcommand{\one}{1}
\newcommand{\Tr}{\mathrm{Tr}^{\vphantom{y}}}

\newcommand{\ad}{\mathrm{Ad}}

\newcommand{\Ch}{\mathsf{Ch}}

\newcommand{\pbwd}{\boldsymbol{m}}
\newcommand{\pbwdd}{\boldsymbol{n}}

\newcommand{\comodul}{{\boldsymbol{a}}}
\newcommand{\coint}{{\boldsymbol{\Lambda}}}
\newcommand{\rint}{{\boldsymbol{\lambda}}}
\newcommand{\balance}{{\boldsymbol{g}}}
\newcommand{\sqs}{{\boldsymbol{u}}}

\newcommand{\ribbon}{{\boldsymbol{v}}}

\newcommand{\hcas}{\boldsymbol{C}}
\newcommand{\hbeta}{\beta}

\newcommand{\cheb}{U}

\numberwithin{equation}{section}
\makeatletter
\@addtoreset{equation}{section}
\@addtoreset{subsubsection}{section}

\def\@secnumfont{\bfseries}
\def\subsubsection{\@startsection{subsubsection}{3}%
  \z@{.5\linespacing\@plus.7\linespacing}{-.5em}%
  {\normalfont\bfseries}}
\def\paragraph{\@startsection{paragraph}{4}%
  \z@\z@{-\fontdimen2\font}%
  \normalfont\bfseries}
\def\subparagraph{\@startsection{subparagraph}{5}%
  \z@\z@{-\fontdimen2\font}%
  \normalfont\bfseries}

\makeatother

\swapnumbers

\newtheorem{lemma}[subsubsection]{Lemma}

\theoremstyle{definition}

\newtheorem{rem}[subsubsection]{Remark}

\begin{document}

\title[Quantum groups in logarithmic CFT]{%
  \vspace*{-\baselineskip}
  Factorizable ribbon quantum groups in logarithmic conformal field
  theories}

\author[Semikhatov]{A.M.~Semikhatov}%

\address{\mbox{}\kern-\parindent Lebedev Physics Institute
  \hfill\mbox{}\linebreak \texttt{ams@sci.lebedev.ru}}

\begin{abstract}
  We review the properties of quantum groups occurring as
  Kazhdan--Lusztig dual to logarithmic conformal field theory models.
  These quantum groups at even roots of unity are not quasitriangular
  but are factorizable and have a ribbon structure; the modular group
  representation on their center coincides with the representation on
  generalized characters of the chiral algebra in logarithmic
  conformal field models.
\end{abstract}

\maketitle

\thispagestyle{empty}

\section{Introduction}
The relation of quantum groups to conformal field theory, discussed
since~\cite{[AGS],[MR],[PS],[MSch]}, has been formulated in the
context of vertex-operator algebras as the Kazhdan--Lusztig
correspondence~\cite{[KL]}.  In a very broad sense (and very roughly),
it states that whenever ``something occurs'' in the representation
category of a vertex-operator algebra, ``something similar occurs'' in
the representation category of an appropriate quantum group; in other
words, there is a functor relating these two categories, although this
functor does not have to be either left- or right-exact.  In this
broad sense, the Kazhdan--Lusztig correspondence is therefore a
\textit{principle} rather than a precise statement; the details of the
functor have to be worked out in each particular case.
For rational conformal field theories, a certain complication follows
from the fact that the chosen vertex-operator-algebra representation
category is semisimple, while the quantum-group one is not, and
additional ``semisimplification'' (taking the quotient over tilting
modules) is needed to ensure the equivalence~\cite{[Fink]}.  But in
logarithmic conformal field theories~\cite{[Gurarie],[GK2]}, the
representation category is already nonsemisimple on the conformal
field theory side, and therefore no quotients need to be a priori
taken on the quantum group side.  The Kazhdan--Lusztig correspondence
extended to the logarithmic realm shows remarkable properties and, in
particular, extends to modular-group
representations~\cite{[FGST],[FGST2],[FGST-q]}.

A more ``physical'' point of view on the Kazhdan--Lusztig
correspondence originates from the observation that screening
operators that commute with a given vertex-operator algebra generate a
quantum group and, moreover, the vertex-operator algebra and the
quantum group are characterized by being each other's commutant,
\begin{equation*}
  [\text{vertex-operator algebra}, \text{quantum group}] = 0,
\end{equation*}
with each of the objects in this relation allowing reconstruction of
the other.  But this picture also applies more as a principle than as
a precise statement,\pagebreak[3] and therefore needs a clarification
as well.  First, the screenings proper generate only the
upper-triangular subalgebra of the quantum group in question; the
entire quantum group has to be reconstructed either by introducing
contour-removal operators (see~\cite{[GRaS]} and the references
therein) or, somewhat more formally, by taking Drinfeld's
double~\cite{[FGST],[FGST3]}.  Second, in seeking the commutant of a
quantum group, it must be specified \textit{where} it is sought, i.e.,
what free-field operators are considered (in particular, what are the
allowed momenta of vertex operators or whether vertex operators are
allowed at all; cf., e.g.,~\cite{[FS-D],[W2n]} in the nonlogarithmic
case).

For several logarithmic conformal field theories, the Kazhdan--Lusztig
correspondence has been shown to have very nice
properties~\cite{[FGST],[FGST2],[FGST3],[FGST-q]}, being ``improved''
compared to the rational case.  Somewhat heuristically, such an
``improvement'' may relate to the fact that the field content in a
logarithmic model is determined not by the cohomology but by the
kernel of the screening(s) (more precisely, by the kernel of a
differential constructed from the screenings; we recall that the
rational models are just the cohomology of such a differential,
cf.~\cite{[Fe],[BMP]}).  Most remarkably, the Kazhdan--Lusztig
correspondence extends to modular group representations.  We recall
that a modular group representation in a logarithmic conformal field
model is generated from the characters $(\chi_a(\tau))$ of the model
by $\modT$- and $\modS$-transformations, the latter being expressed as
\begin{equation}\label{modS-char1}
  \chi_a(-\ffrac{1}{\tau}) = \sum_{b}S_{ab}\chi_b(\tau) +
  \sum_{b'}S'_{ab'}\psi_{b'}(\tau),
\end{equation}
which involves certain functions $\psi_{b'}$, which are not
characters~\cite{[F],[FHST],[FG],[FGST3]}, with
\begin{equation}\label{modS-char2}
  \psi_{a'}(-\ffrac{1}{\tau}) = \sum_{b}S'_{a'b}\chi_b(\tau) +
  \sum_{b'}S'_{a'b'}\psi_{b'}(\tau)
\end{equation}
(the $\chi$ and $\psi$ together can be called generalized, or extended
characters, for the lack of a better name).  On the other hand, in
quantum-group terms, the general theory in~\cite{[Lyu]} (also
see~\cite{[LM],[Kerler]}), which has been developed in an entirely
different context, can be adapted to the quantum groups that are dual
to logarithmic conformal field theories, with the result that a
modular group representation is indeed defined on the quantum group
center.  This representation turns out to be \textit{equivalent} to
the representation generated from the characters.

Another instance where logarithmic conformal field theories and the
corresponding (``dual'') quantum groups show similarity is the fusion
(Verlinde) algebra$/$Grothendieck ring.  The existing data suggest
that the Grothendieck ring of the Kazhdan--Lusztig-dual quantum group
coincides with or ``is closely related to'' the fusion of the chiral
algebra representations on the conformal field theory side.  Two
remarks are in order here: first, comparing a Grothendieck ring with a
fusion algebra implies that the latter is understood ``in a
$K_0$-version,'' when all indecomposable representations are perforce
replaced with direct sums (cf.\ a discussion of this point
in~\cite{[FHST]});\footnote{Whenever indecomposable representations
  are involved, it is of course possible (and more interesting) to
  consider fusion algebras where indecomposable representations are
  treated honestly, i.e., are not replaced by the direct sum of their
  irreducible subquotients~\cite{[GK1],[G-alg]}.  The correspondence
  with quantum groups may also extend from the
  ``$K_0$$/$Grothendieck-style'' fusion to this case (also
  see~\bref{indecomposable} below).}  second, when the logarithmic
conformal field theory has a rational subtheory, the representations
of this rational theory are to be excluded from the comparison (this
is not unnatural though, cf.~\cite{[EF]}).
 
The quantum groups that have so far occurred as dual to logarithmic
conformal field theories are a quantum $s\ell(2)$ and a somewhat more
complicated quantum group, a quotient of the product of two quantum
$s\ell(2)$.  They are dual to logarithmic conformal field theories in
the respective classes of $(p,1)$ and $(\pplus,\pminus)$ models.  In
either case, the Kazhdan--Lusztig-dual quantum group is at an
\textit{even} root of unity.  In either case, the quantum group has a
set of crucial properties, which may therefore be conjectured to be
common to the quantum groups that are dual to logarithmic conformal
models.  These properties and the underlying structures are reviewed
here.  At present, their derivation is only available in a rather
down-to-earth manner, by direct calculation, which somewhat obscures
the general picture.  In what follows, we skip the calculation details
and concentrate on the final results and on the interplay of different
structures associated with the quantum group.

We thus continue the story as seen from the quantum-group side,
following the ideology and results
in~\cite{[FGST],[FGST2],[FGST3],[FGST-q]}.  The necessary excursions
to logarithmic conformal field theory
(see~\cite{[GK2],[GK3],[G-alg],[FFHST],[FHST],[FGST],[FGST3]} and the
references therein) are basically limited to what is needed to
appreciate the similarities with quantum-group structures.  When we
need to be specific (which is almost always the case, because we do
not claim any generality here), we choose the simplest of the two
basic examples, the $\UresSL2$ quantum group dual to the $(p,1)$
logarithmic conformal field theory models, but we indicate the
properties shared by the quantum group $\XXX$ dual to the
$(\pplus,\pminus)$ logarithmic models wherever possible.

The quantum group dual to the logarithmic $(p,1)$ model is $\UresSL2$
at an even root of unity
\begin{equation}\label{the-q}
  \q=\smash[t]{e^{\frac{i\pi}{p}}}
\end{equation}
The three generators $E$, $F$, and $K$ satisfy the relations
\begin{equation}\label{qsl2-comm}
  \begin{gathered}
    KEK^{-1}=\q^2E,\quad
    KFK^{-1}=\q^{-2}F,\\
    [E,F]=\ffrac{K-K^{-1}}{\q-\q^{-1}}
  \end{gathered}
\end{equation}
and the ``constraints''
\begin{equation}\label{qsl2-constr}
  E^{p}=F^{p}=0,\quad K^{2p}=\one.
\end{equation}
We note that Eqs.~\eqref{the-q}--\eqref{qsl2-comm} already imply that
$E^{p}$, $F^{p}$, and $K^{2p}$ are central, which then allows
imposing~\eqref{qsl2-constr} (but $K^p$, which is also central, is
\textit{not} set equal to unity, which makes the difference with a
smaller but more popular version, the so-called \textit{small} quantum
$s\ell(2)$).  As a result, $\UresSL2$ is $2p^3$-dimensional.
\begin{otherqg}
  The quantum group $\XXX$ dual to the $(\pplus,\pminus)$ logarithmic
  model is $2\pplus^3\pminus^3$-dimensional.  We note that the
  ``constraint'' imposed on its Cartan generator is
  $K^{2\pplus\pminus}=\one$.
\end{otherqg}

The Hopf algebra structure of $\UresSL2$ (comultiplication $\Delta$,
counit $\epsilon$, and antipode $S$) is described by
\begin{equation}\label{qsl2-hopf}
  \begin{gathered}
    \Delta(E)=\one\otimes E+E\otimes K,\quad
    \Delta(F)=K^{-1}\otimes F+F\otimes\one,\quad
    \Delta(K)=K\otimes K,\\
    \epsilon(E)=\epsilon(F)=0,\quad\epsilon(K)=1,\\
    S(E)=-EK^{-1},\quad  S(F)=-KF,\quad S(K)=K^{-1}.
  \end{gathered}
\end{equation}
The simplicity of~\eqref{qsl2-comm}--\eqref{qsl2-hopf} is somewhat
misleading.  This quantum group (as well as $\XXX$) has interesting
algebraic properties, the \textit{central} role being played by its
center.

\subsection*{Quantum group center and structures on it}
On the quantum-group side, the main arena of the Kazhdan--Lusztig
correspondence is the quantum group center~$\cZ$.  Of course, it
contains the (``quantum'') Casimir element(s) and the algebra that
they generate, but this does not exhaust the center.

The center carries an $\SLiiZ$ representation, whose
definition~\cite{[Lyu],[LM],[Kerler]} requires three types of
structure: Drinfeld and Radford maps $\Drinfeld$ and $\Radford$, and a
ribbon element~$\ribbon$.  The action of
$\modS=\left(\begin{smallmatrix}
    0&1\\
    -1&0
  \end{smallmatrix}\right)\in\SLiiZ$ on the center is given by
\begin{equation}\label{3-diagram}
  \xymatrix@=16pt{
    {}&\Ch\ar[dl]_{\Drinfeld}
    \ar[dr]^{\Radford}
    \\
    \cZ
    &{}&\cZ\ar@/^/[ll]^{\modS^{-1}}
    }
\end{equation}
where $\Ch$ is the space of $q$-characters (linear functionals
invariant under the coadjoint action), and the action of
$\modT=\left(\begin{smallmatrix}
    1&1\\
    0&1
  \end{smallmatrix}\right)\in\SLiiZ$ 
essentially by (multiplication with) the ribbon element,
\begin{equation}
  \label{eq:T-diagram}
  \cZ\xrightarrow{\quad\ribbon\quad}\cZ.
\end{equation}
(Our definition of $\Radford$ is swapped with its inverse compared to
the standard conventions.)

A possible way to look at the center is to first identify a number of
central elements associated with traces over irreducible
representations and then introduce appropriate pseudotraces.  The
(``quantum'') trace over an irreducible representation gives an
element of $\Ch$, i.e., a functional on the quantum group that is
invariant under the coadjoint representation; these invariant
functionals ($q$-characters) can then be mapped into central elements.
This does \textit{not} cover the entire center.  But then projective
quantum-group modules yield additional $q$-characters, obtained by
taking traces, informally speaking, of nondiagonal components of the
quantum group action, nondiagonal in terms of the filtration of
projective modules.  This gives a basis $\gamma_A$ in the space $\Ch$
of $q$-characters and hence a basis in the center.

Also, the Drinfeld map $\Drinfeld:\Ch\to\cZ$ is an isomorphism of
associative commutative algebras.  Therefore, the center contains (an
isomorphic image of) the Grothendieck ring of the quantum group, which
is thus embedded into a larger associative commutative
algebra.\footnote{That the center contains the image of the
  Grothendieck ring but is larger than it has a counterpart in
  logarithmic conformal field theory, where the set of chiral algebra
  characters $\chi_a$ is to be extended by other functions $\psi_{a'}$
  in order to define a modular group action~\cite{[F],[FHST],[FGST3]}
  and thus, presumably, to construct the space of torus amplitudes
  (also see~\cite{[FG]}).}

In Sec.~\ref{sec:radford}, we review the construction of the Radford
map~$\Radford$.  In Sec.~\ref{sec:modules}, we recall the necessary
facts about the (irreducible and projective) representations of the
relevant quantum groups; their Grothendieck rings are also discussed
there.  In Sec.~\ref{sec:drinfeld}, we recall the $M$[onodromy]
``matrix,'' the Drinfeld map $\Drinfeld$, and the ribbon element.
Together with the Radford map, these serve to define the modular group
action, which we finally consider in Sec.~\ref{sec:modular}.

\section{Radford map and related structures}\label{sec:radford}
We consider the Radford map $\Radford: U^*\to U$; the construction of
$\Radford$ and its inverse involves a cointegral and an integral.

\subsection{Integral and cointegral}
\subsubsection{Integral}
For a Hopf algebra $U$, a \textit{right integral}~$\rint$ is a linear
functional on $U$ satisfying
\begin{equation}\label{int-def}
  (\rint\tensor\id)\Delta(x)=\rint(x)\one\quad\forall x\in U.
\end{equation}
Such a functional exists in a finite-dimensional Hopf algebra and is 
unique up to multiplication~\cite{[Swe]}.

\begin{rem}
  The name \textit{integral} for such a $\rint\in U^*$ is related to
  the fact that~\eqref{int-def} is also the property of a
  right-invariant integral on functions on a group.  Indeed, for a
  function $f$ on a group $G$, $\Delta(f)$ is the function on $G\times
  G$ such that $\Delta(f)(x,y)=f(x y)$, $x,y\in G$.  Then the
  invariance property $\int f(?\,y) = \int f(?)$ can be written as
  $({\int}\,{\tensor}\;\id) \Delta(f) = \int f$.
\end{rem}

\subsubsection{Cointegral} The dual object to $\rint$, an integral for
$U^*$, is sometimes called a cointegral for~$U$.  We give it in the
form needed below, when it is a two-sided cointegral.\footnote{$U^*$
  is therefore assumed unimodular, which turns out to be the case for
  the quantum groups considered below.}\pagebreak[3] A~two-sided
\textit{cointegral}~$\coint$ is an element in $U$ such that
\begin{equation*}
  x\coint=\coint x =\epsilon(x)\coint\quad\forall x\in U.
\end{equation*}
Clearly, the cointegral defines an embedding of the trivial
representation of $U$ into the regular representation.  The
normalization $\rint(\coint)=1$ is typically understood.

\subsection{The Radford map}\label{sec:radford-all}
Let $U$ be a Hopf algebra with a right integral $\rint$ and a
two-sided cointegral~$\coint$.  The Radford map $\Radford:U^*\to U$
and its inverse $\Radford{}^{-1}:U\to U^*$ are given by\footnote{We
  use Sweedler's notation $\Delta(x)=\sum_{(x)}x'\,x''$ (see,
  e.g.,~\cite{[Kassel]}) with the summation symbols omitted in most
  cases; the defining property of the integral, for example, is then
  written as $\rint(x')x''=\rint(x)$.}
\begin{equation}\label{radford-def}
  \Radford(\beta)
  =
  \beta(\coint')\coint'',
  \quad
  \Radford{}^{-1}(x)=\rint(S(x)?).
\end{equation}
\begin{lemma}[\cite{[Swe],[Rad]}]\label{lemma:rad-map}
  $\Radford$ and $\Radford{}^{-1}$ are inverse to each other
  and intertwine the left actions of $U$ on $U$ and $U^*$, and
  similarly for the right actions.
\end{lemma}
Here, the left-$U$-module structure on $U^*$ is given by
$a\acts\beta=\beta(S(a)?)$ (and on $U$, by the regular action).  In
particular, restricting to the space of $q$-characters
(see~\bref{qchar-def}) gives
\begin{equation*}
  \Radford:\Ch\to\cZ.  
\end{equation*}
\begin{proof}\small
  We first establish an invariance property of the integral,
  \begin{equation}\label{rint-rint}
    \rint(xy')y''=\rint(x'y)S^{-1}(x'').
  \end{equation}
  Indeed, $\rint(xy')y''=\rint(x'y')S^{-1}(x''')x''y''
  =\rint((x'y)')S^{-1}(x'')(x'y)''=\rint(x'y)S^{-1}(x'')$.\footnote{Here
    and in what follows, we use the definitions of the antipode and
    counit written in the form (see, e.g., \cite{[Kassel]})
    $x'S(x'')=S(x')x''=\epsilon(x)\one$ and
    $x'\epsilon(x'')=\epsilon(x')x''=x$.  Then, in particular,
    $x'S^{-1}(x''')x''=x$.}  It then follows that
  $\Radford(\Radford^{-1}(x))
  =\Radford(\rint(S(x)?))=\rint(S(x)\coint')\coint''
  \stackrel{\mbox{\tiny by\
      \eqref{rint-rint}}}{=}\rint(S(x)'\coint)S^{-1}(S(x)'')
  =\rint(\epsilon(S(x)')\coint)S^{-1}(S(x)'')
  =S^{-1}(\epsilon(S(x)')S(x)'')=x$.  Similarly, we calculate
  $\Radford^{-1}(\Radford(\beta))
  =\Radford^{-1}(\beta(\coint')\coint'')
  \beta(\coint')\rint(S(\coint'')?)
  =\beta(\rint(S(\coint'')?)\coint')
  =\beta(\rint(S(\coint)'?)S^{-1}(S(\coint)'')) \stackrel{\mbox{\tiny
      by\ \eqref{rint-rint}}}{=}\\
  \beta(\rint(S(\coint)?')?'')
  =\beta(\rint(\coint?')?'')=\beta(\rint(\epsilon(?')\coint)?'')=
  \beta(\epsilon(?')?'')=\beta$.
  
  We next show that $\Radford$ intertwines the left-$U$-module
  structures on $U^*$ and $U$.  With the left-$U$-module structure on
  $U^*$ given by $x\acts\beta=\beta(S(x)?)$, we must prove that
  $\beta(S(x)\coint')\coint'' =x\beta(\coint')\coint''$, or
  $\beta(x\coint')\coint''=S^{-1}(x)\beta(\coint')\coint''$.  But we
  have $\beta(x\coint')\coint''=\beta(\epsilon(x')x''\coint')\coint''
    =\beta((x'\coint)')S^{-1}(x'')(x'\coint)''
    =\beta(\epsilon(x')\coint')S^{-1}(x''')\epsilon(x'')\coint''
    =S^{-1}(x)\beta(\coint')\coint''$.
\end{proof}

\subsection{Traces and the Radford map}\label{traces-radford}
For any irreducible representation $\repX$ of a quantum group $U$, the
(``quantum'') trace in~\eqref{qCh} is an invariant functional on $U$,
i.e., an element of $\Ch(U)$ (see~\bref{qchar-def}).\footnote{The
  reader not inclined to follow the details of the definition
  of~$\balance$ in~\eqref{qCh} may think of it as just the element
  that makes the trace ``quantum,'' i.e., invariant under the
  coadjoint action of the quantum group.}  (But the space of
$q$-characters $\Ch$ is \textit{not} spanned by $q$-traces over
irreducible modules, as we have noted.)  The Radford map sends each of
the $\Tr_{\repX}(\balance^{-1}?)$ functionals into the center $\cZ$
of~$U$:
\begin{equation}
  \Radford:\Tr_{\repX}(\balance^{-1}?)\to\Radford(\repX)\in\cZ.
\end{equation}
It suffices to have $\repX$ range the irreducible representations
of~$U$, because traces ``see'' only irreducible subquotients in
indecomposable representations.  As long as the linear span of
$q$-traces over irreducible modules is not all of the space of
$q$-characters, the Radford-map image of irreducible representations
does not cover the center.

For any central element $a\in\cZ$, its action on an irreducible
representation~$\repX$ is given by multiplication with a scalar, to be
denoted by $a_{\repX}\in\oC$.  By the Radford map properties, we have
the relation 
\begin{equation*}
  a\,\Radford(\repX) = a_{\repX} \Radford(\repX)
\end{equation*}
in the center.  In particular, the Radford-map image of (traces over)
all irreducible representations is the annihilator of the radical in
the center.

\subsubsection{}For $\UresSL2$, it is not difficult to verify that the
right integral and the two-sided cointegral are given by
\begin{gather}
  \rint(F^j E^m K^n)
  =
  \ffrac{1}{\zeta}\,\delta_{j,p-1}\delta_{m,p-1}\delta_{n, p+1},\\
  \label{coint}
  \coint=\zeta\,F^{p-1}E^{p-1}\sum_{j=0}^{2p-1}K^j,
\end{gather}
where we choose the normalization factor $\zeta=\sqrt{\ffrac{p}{2}}\,
\ffrac{1}{([p-1]!)^2}$~\cite{[FGST]}.\footnote{In general, the
  (co)integral is defined up to a nonzero factor, but factorizable
  ribbon quantum groups offer a ``canonical'' normalization, derived
  from the condition $\modS^2=\id$ on the center; in accordance
  with~\eqref{3-diagram}, the normalization of $\modS$ is inherited
  from the normalization of $\Radford$, and hence from that of the
  cointegral.}
\begin{otherqg}{}
  For $\XXX$, the expressions for $\rint$ and $\coint$
  in~\cite{[FGST-q]} also hinge on the fact that $p-1$ is the highest
  nonzero power of the off-diagonal quantum group generators.
\end{otherqg}

\subsection{Comodulus}\label{sec:comodulus} Another general notion
that we need is that of a comodulus.  For a right integral $\rint$,
the \textit{comodulus} ``measures'' how much $\rint$ differs from a
left integral (see~\cite{[Rad-min]}): it is an element $\comodul\in U$
such that
\begin{equation*}
  (\id\tensor\rint)\Delta(x)
  =\rint(x)\comodul\quad\forall x\in U.
\end{equation*}

A simple calculation then shows that the $\UresSL2$ comodulus is
$\comodul=K^2$.
\begin{otherqg}{}
  For $\XXX$, the comodulus is $\comodul=K^{2\pplus - 2\pminus}$.
\end{otherqg}

\section{Quantum group modules: from irreducible to
  projective}\label{sec:modules}
Irreducible (simple) and projective quantum group representations are
considered below.  By general philosophy of the Kazhdan--Lusztig
duality, the irreducible quantum-group representations somehow
``correspond'' to irreducible chiral algebra representations in
logarithmic conformal models.  In particular, the Grothendieck ring is
generally related to fusion in conformal field theory.  While direct
calculation of the fusion of chiral algebra representations is
typically quite difficult, this Grothendieck-ring structure may be
considered a poor man's fusion (there is evidence that it is not
totally meaningless).  Apart from irreducible representations, their
projective covers play an important role.  To be specific, we now
describe some aspects of the representation theory in the example
of~$\UresSL2$.

\subsection{Irreducible representations and the Grothendieck ring}
\label{subsec:irrep}
There are $2p$ irreducible $\UresSL2$-representations $\repX^{\pm}_r$,
which can be conveniently labeled by the $\pm$ and $1\leq r\leq p$.
The highest-weight vector $\vectv{ r}^{\pm}$ of $\repX^{\pm}_r$ is
annihilated by $E$, and its weight is determined by $K\vectv{r}^{\pm}
= \pm \q^{r - 1} \vectv{ r}^{\pm}$.  The representation dimensions are
$\dim\repX^{\pm}_r=r$.
Some readers might find it suggestive to visualize the representations
$\repX^\pm_r$ arranged into a ``Kac table,'' a single row of boxes
labeled by $r=1,\dots,p$, each carrying a ``$+$'' and a ``$-$''
representation: $\underbrace{\boxed{~\rule{0pt}{3.6pt}}\kern-.5pt
  \boxed{~\rule{0pt}{3.6pt}}\kern-.5pt\boxed{~\rule{0pt}{3.6pt}}\dots
  \boxed{~\rule{0pt}{3.6pt}}\kern-.5pt\boxed{~\rule{0pt}{3.6pt}}}_{p}$.

We next recall that the Grothendieck ring is the free Abelian group
generated by symbols $[M]$, where $M$ ranges over all representations
subject to relations $[M]=[M']+ [M'']$ for all exact sequences $0 \to
M'\to M\to M''\to 0$.  Multiplication in the ring is induced by the
tensor product of representations, with any indecomposable module
occurring in the tensor product replaced by a sum of its simple
subquotients.

\mypar{$\UresSL2$} The Grothendieck ring of $\UresSL2$ is (rather
straightforwardly~\cite{[FGST]}) found to be given by
\begin{equation}\label{the-fusion}
  \repX^{\alpha}_{r}\repX^{\alpha'}_{s}
  =\smash[t]{\sum_{\substack{t=|r - s| + 1\\
        \mathrm{step}=2}}^{r + s - 1}}
    \tilde\repX^{\alpha\alpha'}_{t}
\end{equation}
where
\begin{equation*}
  \tilde\repX^{\alpha}_{r}
  =
  \begin{cases}
    \repX^{\alpha}_{r},&1\leq r\leq p,\\
    \repX^{\alpha}_{ 2p - r} + 2\repX^{-\alpha}_{ r - p}, & p + 1 \leq r
    \leq 2p - 1.
  \end{cases}
\end{equation*}
It can also be described in terms of Chebyshev polynomials, as the
quotient of the polynomial ring $\oC[x]$ over the ideal generated by
the polynomial
\begin{equation*}
  \hat\Psi_{2p}(x)=
  \cheb_{2p+1}(x)-\cheb_{2p-1}(x)-2,
\end{equation*}
where $\cheb_s(x)$ are Chebyshev polynomials of the second kind:
\begin{equation*}
  \cheb_s(2\cos t)=\ffrac{\sin s t}{\sin t},\quad s\geq1.
\end{equation*}
They satisfy the recurrence relations
$x\cheb_s(x)=\cheb_{s-1}(x)+\cheb_{s+1}(x)$, $s \geq 2$, with the
initial data $\cheb_1(x)=1$, $\cheb_2(x)=x$. \ Moreover, let
\begin{equation}\label{basis-P}
  P_s(x)=
  \begin{cases}
    \cheb_s(x),& 1\leq s\leq p,\\
    \half\cheb_s(x)-\half\cheb_{2p-s}(x),& p+1\leq s\leq 2p.
  \end{cases}
\end{equation}
Under the quotient map, the image of each $P_s$ coincides with
$\repX^+_{s}$ for $1\leq s\leq p$ and with $\repX^-_{s-p}$ for
$p+1\leq s\leq2p$.

The algebra in~\eqref{the-fusion} is a nonsemisimple Verlinde algebra
(commutative associative algebra with nonnegative integer structure
coefficients, see~\cite{[Fuchs]}), with a unit given by~$\repX^+_1$.
The algebra contains the ideal~$\mathsf{V}_{p+1}$ generated by
$\repX^{+}_{p-r}+\repX^{-}_{r}$ with $1\leq r\leq p-1$,
$\repX^{+}_{p}$, and~$\repX^{-}_{p}$.  The quotient over
$\mathsf{V}_{p+1}$ is a semisimple Verlinde algebra and in fact
coincides with the fusion of the unitary $\hSL2$ representations of
level~$p-2$.\footnote{It may be worth emphasizing that a Verlinde
  algebra structure involves not only an associative commutative
  structure but also a distinguished basis (the above quotient is that
  of Verlinde algebras).  In particular, the reconstruction of the
  Verlinde algebra from its block decomposition as an associative
  algebra (the structure of primitive idempotents and elements in the
  radical in the algebra) requires extra information,
  cf.~\cite{[FHST]}.}

The same algebra was derived in~\cite{[FHST]} from modular
transformations of the triplet $W$-algebra characters in logarithmic
$(p,1)$-models within a nonsemisimple generalization of the Verlinde
formula (also see~\cite{[FK]} for comparison with other derivations).

\mypar{$\XXX$}
\begin{otherqg}
  The quantum group $\XXX$ dual to the $(\pplus,\pminus)$ logarithmic
  model has $2\pplus\pminus$ irreducible representations
  $\repX^{\pm}_{r,r'}$, $1\leq r\leq\pplus$, $1\leq r'\leq\pminus$,
  with $\dim\repX^{\pm}_{r,r'}=r r'$.  They can be considered arranged
  into a ``Kac table''
  \raisebox{-14pt}{${}^{p'}\left\{\rule{0pt}{21pt}\right.$}\kern-6pt
  $\begin{array}[t]{l}
    \boxed{~\rule{0pt}{3.6pt}}\kern-.5pt\boxed{~\rule{0pt}{3.6pt}}
    \kern-.5pt\boxed{~\rule{0pt}{3.6pt}}\dots\boxed{~\rule{0pt}{3.6pt}}
    \kern-.5pt\boxed{~\rule{0pt}{3.6pt}}\\[-7pt]
    \makebox[63.6pt]{\dotfill}\\[-5pt]
    \boxed{~\rule{0pt}{3.6pt}}\kern-.5pt\boxed{~\rule{0pt}{3.6pt}}
    \kern-.5pt\boxed{~\rule{0pt}{3.6pt}}
    \dots\boxed{~\rule{0pt}{3.6pt}}
    \kern-.5pt\boxed{~\rule{0pt}{3.6pt}}\\[-5pt]
    \underbrace{\boxed{~\rule{0pt}{3.6pt}}
      \kern-.5pt\boxed{~\rule{0pt}{3.6pt}}
      \kern-.5pt\boxed{~\rule{0pt}{3.6pt}}\dots
      \boxed{~\rule{0pt}{3.6pt}} \kern-.5pt
      \boxed{~\rule{0pt}{3.6pt}}}_{p}
  \end{array}$, with each box carrying a ``$+$'' and a ``$-$''
  representation.

  The Grothendieck ring structure
  is given by~\cite{[FGST-q]}
  \begin{equation}\label{the-Groth}
    \repX^{\alpha}_{r,r'}\repX^{\beta}_{s,s'}
    =\sum_{\substack{u=|r - s| + 1\\
        \mathrm{step}=2}}^{r + s - 1}
    \sum_{\substack{u'=|r' - s'| + 1\\
        \mathrm{step}=2}}^{r' + s' - 1}
    {\tXX}^{\alpha\beta}_{u,u'},
  \end{equation}
  where
\begin{equation*}
    {\tXX}^{\alpha}_{r,r'} =
    \begin{cases}
      \repX^{\alpha}_{r,r'},&
        1\leq r\leq \pplus,\
        1\leq r'\leq \pminus,
      \\[2pt]
      \repX^{\alpha}_{2\pplus - r,r'} + 2\repX^{-\alpha}_{r - \pplus,
        r'},&
        \pplus\!+\!1\leq r\leq 2 \pplus\!-\!1,\
        1\leq r'\leq \pminus,
      \\[2pt]
      \repX^{\alpha}_{r,2\pminus - r'} + 2\repX^{-\alpha}_{r,r' -
        \pminus},&
        1\leq r\leq \pplus,\
        \pminus\!+\!1\leq r'\leq 2 \pminus\!-\!1,
      \\[2pt]
      \mbox{}\kern-3pt\begin{aligned}[b] &\repX^{\alpha}_{2\pplus - r, 2
          \pminus - r'}
        + 2\repX^{-\alpha}_{2\pplus - r, r' - \pminus}\\
        &{}+ 2\repX^{-\alpha}_{r - \pplus, 2 \pminus - r'} +
        4\repX^{\alpha}_{r - \pplus, r' - \pminus},
      \end{aligned}&
        \pplus\!+\!1\leq r\leq 2 \pplus\!-\!1,\
        \pminus\!+\!1\leq r'\leq 2 \pminus\!-\!1.
    \end{cases}
  \end{equation*}
  This algebra is a quotient of \ $\oC[x,y]$ as described
  in~\cite{[FGST-q]}.
  The radical in this nonsemisimple Verlinde algebra (with a unit
  given by $\repX^+_{1,1}$) is generated by the algebra action on
  $\repX^+_{p,p'}$; the quotient over the radical coincides with the
  fusion of the $(p,p')$ Virasoro minimal model.

  The above algebra is a viable candidate for the (``$K_0$-type'')
  fusion of $W$-algebra representations in the logarithmic
  $(p,p')$-models (see~\cite{[FGST3],[EF]}).
\end{otherqg}

\subsection{Indecomposable modules}
\subsubsection{}\label{sec:indecomp}
Irreducible quantum-group modules can be ``glued'' together to produce
indecomposable representations.\pagebreak[3] Already for $\UresSL2$,
its indecomposable representations (which have been classified, rather
directly, in~\cite{[FGST2]} or can be easily deduced from a more
general analysis in~\cite{[Erd]}) are rather numerous.  Apart from the
projective modules, to be considered separately
in~\bref{subsec:proj-mod}, indecomposable representations are given by
families of modules $\mathscr{W}^{\pm}_r(n)$,
$\mathscr{M}^{\pm}_r(n)$, and $\mathscr{O}^{\pm}_r(n,z)$ that can be
respectively represented as
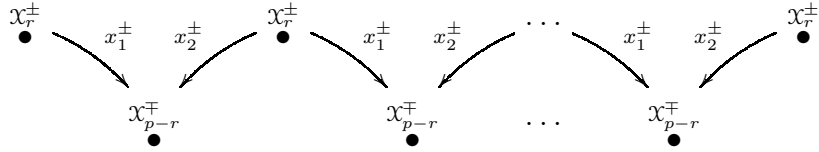
\begin{equation*}
  \xymatrix@=12pt{
    \stackrel{\repX^{\pm}_{r}}{\bullet}\ar@/^/[dr]^{x^\pm_1}&
    &\stackrel{\repX^{\pm}_{r}}{\bullet}\ar@/_/[dl]_{x^\pm_2}
    \ar@/^/[dr]^{x^\pm_1}&&\dots\ar@/_/[dl]_{x^\pm_2}\ar@/^/[dr]^{x^\pm_1}&
    &\stackrel{\repX^{\pm}_{r}}{\bullet}
    \ar@/_/[dl]_{x^\pm_2}\\
    &\stackrel{\quad\repX^{\mp}_{p-r}\quad}{\bullet}&&
    \stackrel{\quad\repX^{\mp}_{p-r}\quad}{\bullet}&\dots&
    \stackrel{\quad\repX^{\mp}_{p-r}\quad}{\bullet}&
  }
\end{equation*}

\noindent
(with $1\leq r \leq p\,{-}\,1$, and integer $n\geq2$ the number of
$\repX^{\pm}_{r}$ modules),
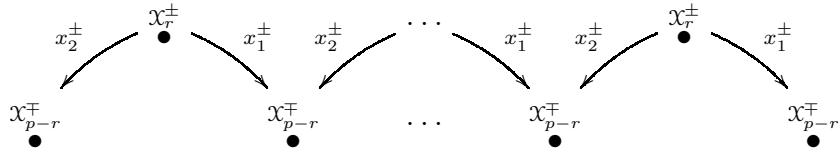
\begin{equation*}
  \xymatrix@=12pt{
    &\stackrel{\repX^{\pm}_{r}}{\bullet}\ar@/_/[dl]_{x^\pm_2}
    \ar@/^/[dr]^{x^\pm_1}&&\dots\ar@/_/[dl]_{x^\pm_2}\ar@/^/[dr]^{x^\pm_1}&
    &\stackrel{\repX^{\pm}_{r}}{\bullet}
    \ar@/_/[dl]_{x^\pm_2}\ar@/^/[dr]^{x^\pm_1}&\\
    \stackrel{\quad\repX^{\mp}_{p-r}\quad}{\bullet}&&
    \stackrel{\quad\repX^{\mp}_{p-r}\quad}{\bullet}&\dots&
    \stackrel{\quad\repX^{\mp}_{p-r}\quad}{\bullet}&&
    \stackrel{\quad\repX^{\mp}_{p-r}\quad}{\bullet}
  }
\end{equation*}
(with $1\leq r \leq p\,{-}\,1$, and integer $n\geq2$ the number of
$\repX^{\mp}_{p-r}$ modules), and
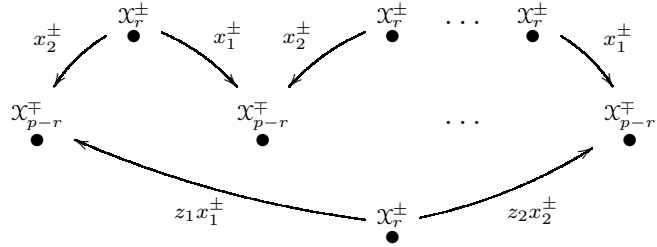
\begin{equation*}
  \xymatrix@=12pt{
    &\stackrel{\repX^{\pm}_{r}}{\bullet}
    \ar@/_/[dl]_{x^\pm_2}\ar@/^/[dr]^{x^\pm_1}&
    &\stackrel{\repX^{\pm}_{r}}{\bullet}
    \ar@/_/[dl]_{x^\pm_2}
    &{\mbox{}\kern-20pt\dots\kern-20pt\mbox{}}
    &\stackrel{\repX^{\pm}_{r}}{\bullet}
    \ar@/^/[dr]^{x^\pm_1}
    \\
    \stackrel{\repX^{\mp}_{p-r}}{\bullet}&&
    \stackrel{\quad\repX^{\mp}_{p-r}\quad}{\bullet}&
    &{\mbox{}\kern-20pt\dots\kern-20pt\mbox{}}
    &&\stackrel{\repX^{\mp}_{p-r}}{\bullet}&
    \\
    &&&\stackrel{\repX^{\pm}_{r}}{\bullet}
    \ar@/_/[urrr]_{z_2x^\pm_2}\ar@/^/[ulll]^{z_1x^\pm_1}&
    &&
  }
\end{equation*}
(with $1\leq r \leq p\,{-}\,1$, $z=z_1:z_2\in\oC\oP^1$, and integer
$n\geq1$ the number of the $\repX^{\pm}_{r}$ modules).  The small
$x^+_i$ and $x^-_i$, $i=1,2$, are basis elements chosen in the
respective spaces
$\oC^2=\Ext(\repX^{+}_{r},$\linebreak[0]$\repX^{-}_{p-r})$ and
$\oC^2=\Ext(\repX^{-}_{p-r},\repX^{+}_{r})$; they in fact generate the
algebra $\EXT_s$ (with the Yoneda product) with the relations
\begin{equation*}
  x^+_i x^+_j=x^-_ix^-_j=x^+_1x^-_2+x^+_2x^-_1=x^-_1x^+_2+x^-_2x^+_1=0
\end{equation*}
(see~\cite{[FGST2]} for the details).

Interestingly, a very similar picture (the ``zigzag,'' although not
the ``$\mathscr{O}$'' modules) also occurred in a different
context~\cite{[GQS],[QS]}.

\subsubsection{}
The representation category decomposes into subcategories as follows.
For $\UresSL2$, the familiar (``quantum'') Casimir element
\begin{equation}\label{hCas}
  \hcas=(\q-\q^{-1})^2 EF + \q^{-1}K+\q K^{-1}
\end{equation}
satisfies the minimal polynomial relation $\Psi_{2p}(\hcas)=0$,
where~\cite{[FGST]}
\begin{equation}\label{Psi}
  \Psi_{2p}(x) =
  (x-\hbeta_0)\,(x-\hbeta_p)
  \prod_{s=1}^{p-1}(x-\hbeta_s)^2, \quad
  \hbeta_s=\q^s+\q^{-s}.
\end{equation}
This relation yields a decomposition of the representation category
into the direct sum of full subcategories $\lc(s)$ such that
$(\hcas-\hbeta_s)$ acts nilpotently on objects in $\lc(s)$.  Because
$\hbeta_s\neq \hbeta_{s'}$ for $0\leq s\neq s'\leq p$, there
are~$p\,{+}\,1$ full subcategories~$\lc(s)$ for $0\,{\leq}\,
s\,{\leq}\, p$.  Each $\lc(s)$ with $1 \leq s \leq p{-}1$ contains
precisely two irreducible modules $\repX^{+}_{s}$ and
$\repX^{-}_{p-s}$ (because the Casimir element acts by multiplication
with $\hbeta_s$ on precisely these two) and infinitely many
indecomposable modules.  The irreducible modules $\repX^{+}_{p}$ and
$\repX^{-}_{p}$ corresponding to the respective eigenvalues $\hbeta_p$
and $\hbeta_0$ comprise the respective categories $\lc(p)$ and
$\lc(0)$.

\subsection{Projective modules}\label{subsec:proj-mod}
The process of constructing the extensions stops at projective
modules\,---\,projective covers of each irreducible representation.
Taking direct sums of projective modules then gives projective covers
of all indecomposable representation.

A few irreducible representations are their own projective covers;
these are $\repX^{\pm}_{p}$ for $\UresSL2$
\begin{otherqg}{}
  and $\repX^{\pm}_{p,p'}$ for $\XXX$.
\end{otherqg}
The other irreducible representations have projective covers filtered
by several irreducible subquotients.

For $\UresSL2$, the projective cover $\projP^\pm_r$ of
$\repX^\pm_{r}$, $r=1,\dots,p-1$, can be represented as\footnote{In
  diagrams of this type, first, the arrows are directed towards
  submodules; second, it is understood that the quantum group action
  on each irreducible representation is changed in agreement with the
  arrows connecting a given subquotient with others.  This is of
  course true for the ``two-floor'' indecomposable modules considered
  above, but is even more significant for the projective modules,
  where the $\overset{\repX^{\pm}_{r}}{\bullet}{}\longrightarrow{}
  \overset{\!\!\!\repX^{\mp}_{p-r}\!\!\!}{\bullet}$ extensions alone
  do not suffice to describe the quantum group action.  Constructing
  the quantum group action there requires some more work, but is not
  very difficult for each of the quantum groups considered here, as
  explicit formulas in~\cite{[FGST],[FGST-q]} show.}
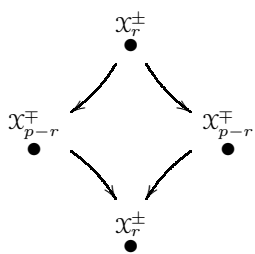
\begin{equation}\label{schem-proj}
  \xymatrix@=12pt{
    &&\stackrel{\repX^{\pm}_{r}}{\bullet}
    \ar@/^/[dl]
    \ar@/_/[dr]
    &\\
    &\stackrel{\repX^{\mp}_{p{-}r}}{\bullet}\ar@/^/[dr]
    &
    &\stackrel{\repX^{\mp}_{p{-}r}}{\bullet}\ar@/_/[dl]
    \\
    &&\stackrel{\repX^{\pm}_{r}}{\bullet}&
  }
\end{equation}
It follows that $\dim\projP^\pm_{r}=2p$.
\begin{otherqg}
  For $\XXX$, besides $2$ irreducible projective modules of dimension
  $\pplus\pminus$, there are $2(\pplus-1+\pminus-1)$ projective
  modules of dimension $2\pplus\pminus$ and $2(\pplus-1)(\pminus-1)$
  projective modules of dimension $4\pplus\pminus$
  (see~\cite{[FGST-q]}, where a diagram with 16 subquotients is also
  given).
\end{otherqg}

Regarding this picture for projective modules (as well as more
involved pictures in~\cite{[FGST-q]}), it is useful to keep in mind
that because of the periodicity in powers of $\q$, the top and the
bottom subquotients sit in the same grade (measured by eigenvalues of
the Cartan generator $K$), as do the two ``side'' subquotients.  A
picture that makes this transparent and which shows the states in
projective modules can be drawn as follows.  Taking $\UresSL2$ with
$p=5$ and choosing $r=3$ for example, we first represent the
$\repX^-_{r}=\repX^-_{3}$ and $\repX^+_{p-r}=\repX^+_{2}$ irreducible
modules as
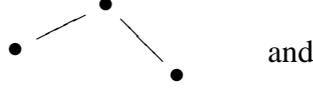
\def\variousvert{\ifcase\xypolynode\or\mbox{}\or\mbox{}\or\,\;
  \bullet\or\,\;\bullet\or\,\;\bullet\fi}
\mbox{}\qquad$\smash{\raisebox{-40pt}{\xy /l1.2pc/:
    \xypolygon5"B"{~*{\variousvert}~><{@{}} ~:{(2,1):}}
    \ar@{-}"B5";"B4" \ar@{-}"B4";"B3" \endxy}}$
\qquad\raisebox{-30pt}{and}\qquad
\def\variousvert{\ifcase\xypolynode\or\,\;\bullet\or
  \,\;\bullet\or\mbox{}\or\mbox{}\or\mbox{}\fi} \qquad$\smash{\xy
  /l1.2pc/: \xypolygon5"B"{~*{\variousvert}~><{@{}} ~:{(2,1):}}
  \ar@{-}"B2";"B1"
  \endxy}$\\[12pt]
and then construct their extension\\[-8pt]
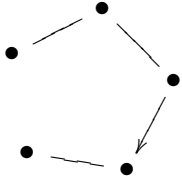
\begin{minipage}[t]{.2\linewidth}
  \mbox{}\kern-30pt
  \mbox{}\qquad \xy /l1.2pc/: \xypolygon5"B"{~*{\,\;\bullet}~><{@{}}
    ~:{(2,1):}} \ar@{-}"B5";"B4" \ar@{-}"B4";"B3" \ar@{-}"B2";"B1"
  \ar@{->}"B3";"B2"
  \endxy
\end{minipage}\kern-10pt
\begin{minipage}[t]{.55\linewidth}
  \vspace*{-2\baselineskip} which actually gives a Verma module (the
  arrow is directed to the submodule).  {}From this module and a
  contragredient one, we further construct the projective module
  $\modL_{2}$ as the extension
\end{minipage}
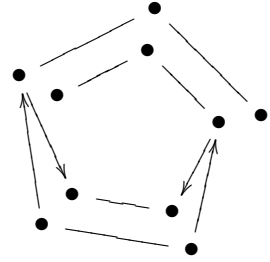
\begin{minipage}[b]{.2\linewidth}
\mbox{}\qquad\xy /l1.2pc/: {\xypolygon5"A"{~*{\,\;\bullet}~><{@{}}
    ~:{(3,1.5):}}}, \xypolygon5"B"{~*{\,\;\bullet}~><{@{}} ~:{(2,1):}}
\ar@{-}"B5";"B4" \ar@{-}"B4";"B3" \ar@{-}"B2";"B1" \ar@{->}"B3";"B2"
\ar@{-}"A5";"A4" \ar@{-}"A4";"A3" \ar@{-}"A2";"A1" \ar@{->}"A1";"A5"
\ar@{->}"A5";"B1" \ar@{->}"A2";"B3"
\endxy
\mbox{}
\end{minipage}\\
where pairs of nearby dots represent states that actually sit in the
same grade.

\subsubsection{From the Grothendieck ring to the tensor algebra}
The results in~\cite{[Erd]} go beyond the Grothendieck ring for the
quantum group closely related to $\UresSL2$: tensor products of the
indecomposable representations are evaluated there.  It follows
from~\cite{[Erd]} that the $\UresSL2$ Grothendieck
ring~\eqref{the-fusion} is in fact the result of ``forceful
semisimplification'' of the following tensor product algebra of
irreducible representations.  First, if $r + s - p\leq1$, then,
obviously, only irreducible representations occur in the
decomposition:
\begin{align*}
  \repX^{\alpha}_{r}\tensor\repX^{\beta}_{s}
  &=\bigoplus_{\substack{t=|r - s| + 1\\
      \mathrm{step}=2}}^{r + s - 1}
  \repX^{\alpha\beta}_{t}\\
  \intertext{(the sum contains $\min(r,s)$ terms).  Next, if
    $r+s-p\geq 2$ and is even, $r+s-p=2n$ with $n\geq1$, then}
  \repX^{\alpha}_{r}\tensor\repX^{\beta}_{s}
  &=\bigoplus_{\substack{t=|r - s| + 1\\
      \mathrm{step}=2}}^{2p-r-s-1} \repX^{\alpha\beta}_{t} \oplus
  \bigoplus_{a=1}^{n}\projP^{\alpha\beta}_{p+1-2a}.\\
  \intertext{Finally, if $r+s-p\geq 3$ and is odd, $r+s-p=2n+1$ with
    $n\geq1$, then} \repX^{\alpha}_{r}\tensor\repX^{\beta}_{s}
  &=\bigoplus_{\substack{t=|r - s| + 1\\
      \mathrm{step}=2}}^{2p-r-s-1} \repX^{\alpha\beta}_{t} \oplus
  \bigoplus_{a=0}^{n}\projP^{\alpha\beta}_{p-2a}.
\end{align*}
We note that in each of the last two formulas, the first sum in the
right-hand side contains $p-\max(r,s)$ terms, and therefore disappears
whenever $\max(r,s)=p$ (see~\cite{[Erd]} for the tensor products of
other modules in~\bref{sec:indecomp}).

\subsubsection{Remarks}\label{indecomposable}
\begin{enumerate}
\item It follows that the irreducible $\UresSL2$ representations
  produce only (themselves and) projective modules in the tensor
  algebra.  Because tensor products of any modules with projective
  modules decompose into projective modules, we can consistently
  restrict ourself to only the irreducible and projective modules (in
  other words, there is a subring in the tensor algebra).  This is a
  very special situation, however, specific to $\UresSL2$ (and the
  slightly larger algebra in~\cite{[Erd]}); generically,
  indecomposable representations other than the projective modules
  occur in tensor products of irreducible representations.\footnote{I
    thank V.~Schomerus for this remark and a discussion of this
    point.}
  \begin{otherqg}{}
    In particular, the true tensor algebra behind the Grothendieck
    ring in~\eqref{the-Groth} is likely to involve various other
    indecomposable modules in the product of irreducible
    representations.
  \end{otherqg}
  Already for $\UresSL2$, specifying the \textit{full} tensor algebra
  means evaluating the products of all the representations listed
  in~\bref{sec:indecomp}.

\item Reiterating the point in~\cite{[FGST2]}, we note that the
  previous remark fully applies to fusion of the chiral algebra
  (triplet $W$-algebra~\cite{[K-first],[GK2],[GK3]}) representations
  in $(p,1)$ logarithmic conformal field theory models, once the
  fusion is taken not in the $K_0$-version but with ``honest''
  indecomposable representations~\cite{[GK1],[G-alg]}.  While it is
  possible to consider such a fusion of only irreducible and
  projective $W$-algebra modules, the full fusion algebra must include
  all of the ``$\mathscr{W}$,'' ``$\mathscr{M}$,'' and
  ``$\mathscr{O}$'' indecomposable modules of the triplet algebra
  (with the last ones, somewhat intriguingly, being dependent on
  $z\in\oC\oP^1$).
\end{enumerate}

\subsubsection{Pseudotraces} Projective modules serve another,
somewhat technical but useful purpose.  It was noted
in~\bref{traces-radford} that traces over irreducible representations
do not span the entire space $\Ch$ of $q$-characters.  Projective
modules provide what is missing: they allow constructing pseudotraces
$\Tr_{\mathbb{P}}(\balance^{-1} ? \sigma)$ (for certain maps and
modules $\sigma:\mathbb{P}\to\mathbb{P}$) that together with the
traces $\Tr_{\repX}(\balance^{-1}?)$ over irreducible representations
span all of $\Ch$: a basis $\gamma_A$ in $\Ch$ can be constructed such
that with a subset of the $\gamma_A$ is given by traces over
irreducible representations and the rest by pseudotraces associated
with projective modules in each full subcategory.

The strategy for constructing the pseudotraces is as follows.  For any
(reducible) module $\mathbb{P}$ and a map
$\sigma:\mathbb{P}\to\mathbb{P}$, the functional
\begin{equation}
 \gamma{}:{}
  x\mapsto \Tr_{\mathbb{P}}(\balance^{-1} x \sigma)
\end{equation}
is a $q$-character if and only if (cf.~\eqref{Ch-def})
\begin{equation}\label{q-char-cond}
    0=\gamma(x y)- \gamma(S^2(y)x) \equiv
    \Tr_{\mathbb{P}}
    (\balance^{-1}x[y,\sigma]).
\end{equation}
It is possible to find reducible indecomposable modules $\mathbb{P}$
and maps $\sigma$ satisfying~\eqref{q-char-cond}.  This requires
taking $\mathbb{P}$ to be the projective module in a chosen full
subcategory (one of those containing more than one module).  For
$\UresSL2$, this is
\begin{equation}\label{PP-p1}
  \mathbb{P}_r=\projP^+_r\oplus\projP^-_{p-r},\quad 1\leq r\leq p-1,
\end{equation}
\begin{otherqg}{}
  and for $\XXX$ this is the direct sum
  \begin{equation}\label{PP}
    \mathbb{P}_{r,r'}
    =\projP^{+}_{r,r'}\oplus\projP^{-}_{\pplus-r,r'} 
    \oplus\projP^{-}_{r,\pminus-r'}\oplus\projP^{+}_{\pplus-r,\pminus-r'},
  \end{equation}
  plus the ``boundary'' cases where either $r=p$ or $r'=p'$, with two
  terms in the sum (here, $\projP^\pm_{r,r'}$ is the projective cover
  of the irreducible representation~$\repX^\pm_{r,r'}$).
\end{otherqg}
In all cases, $\sigma$ is a linear map that sends the bottom module in
the filtration of each projective module into ``the same'' module at a
higher level in the filtration.  Such maps are not defined uniquely
(e.g., they depend on the choice of the bases and, obviously, on the
``admixture'' of lower-lying modules in the filtration), but anyway,
taken together with the traces over irreducible representations, they
allow constructing a basis in $\Ch$.  For $\UresSL2$, there is a
single pseudotrace for each $r$ in~\eqref{PP-p1} obtained by letting
$\sigma$ send the bottom of both diamonds of type~\eqref{schem-proj}
into the top.  This gives just $p\,{-}\,1$ linearly independent
elements of~$\Ch$.
\begin{otherqg}{}
  The structure for $\XXX$ is somewhat richer, and the counting goes
  as follows~\cite{[FGST-q]}.  There are not one but three other
  copies of the bottom subquotient in each of the four projective
  modules in~\eqref{PP}.  For the $12$ parameters thus emerging, $7$
  constraints follow from~\eqref{q-char-cond}.  Of the remaining $5$
  different maps satisfying~\eqref{q-char-cond}, there is just one
  (the map to the very top) for each of the
  $\half(p\,{-}\,1)(p'\,{-}\,1)$ modules of form~\eqref{PP}, two for
  each of the $\half(p\,{-}\,1)p'$ modules, and two more for each of
  the $\half p(p'\,{-}\,1)$ modules.  This gives the total of
  \begin{equation*}
    \half(p-1)(p'-1) + (p-1)p' + p(p'-1)
  \end{equation*}
  linearly independent pseudotraces.  Together with the traces over
  $2\pplus\pminus$ irreducible representations, we thus obtain
  $\half(3\pplus\,{-}\,1)(3\pminus\,{-}\,1)$ linearly independent
  elements of~$\Ch$%
\end{otherqg}.

\subsubsection{``Radford'' basis}Radford-map images of the basis
$\gamma_A$ of traces and pseudotraces in $\Ch$ give a basis
\begin{equation*}
  \smash[t]{\Radford_A=\Radford(\gamma_A)}
\end{equation*}
in the quantum-group center~$\cZ$.  This basis plays an important role
in what follows, being one of the two special bases related by
$\modS\in\SLiiZ$.  The other special basis is associated with the
Drinfeld map considered in the next section.

\subsubsection{Projective modules and the center} Projective modules
are also a crucial ingredient in finding the quantum group center.
Central elements are in a $1:1$ correspondence with \textit{bimodule}
endomorphism of the regular representation.  We recall that viewed as
a left module, the regular representation decomposes into projective
modules, each entering with the multiplicity given by the dimension of
its simple quotient.  Generalizing this picture to a bimodule
decomposition shows that the multiplicities are in fact tensor factors
with respect to the right action.  A typical block of the bimodule
decomposition of the regular representation looks as follows: with
respect to the left action, it is a sum of projective modules in one
full subcategory, with each projective (externally) tensored with a
suitable simple module.  For $\UresSL2$, where the subquotients are
few and therefore the picture is not too complicated, it can be drawn
as~\cite{[FGST]}

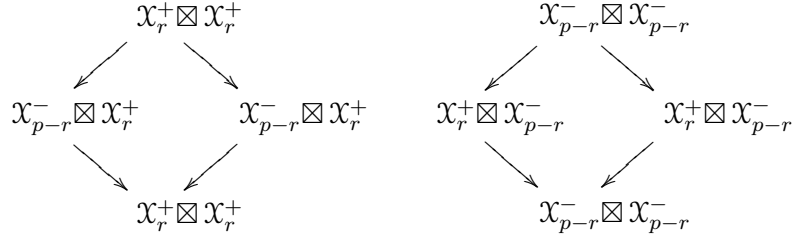
\begin{small}\vspace*{-.8\baselineskip}
  \begin{equation*}
    \mbox{}\kern-62pt
    \xymatrix@=16pt{%
      *{}&{\repX^+_{r}\makebox[0pt][l]{${\boxtimes}\,\repX^+_{r}$}}
      \ar[1,-1]
      \ar[1,1]
      &*{}&*{}&*{}
      *{}&{\repX^-_{p{-}r}\makebox[0pt][l]{${\boxtimes}\,\repX^-_{p{-}r}$}}
      \ar[1,-1]
      \ar[1,1]
      \\
      {\repX^-_{p{-}r}\makebox[0pt][l]{${\boxtimes}\,\repX^+_{r}$}}
      \ar[1,1]
      &*{}&
      {\repX^-_{p{-}r}\makebox[0pt][l]{${\boxtimes}\,\repX^+_{r}$}}
      \ar[1,-1]
      &*{\quad}
      &{\repX^+_{r}\makebox[0pt][l]{${\boxtimes}\,\repX^-_{p{-}r}$}}
      \ar[1,1]
      &*{}
      &{\repX^+_{r}\makebox[0pt][l]{${\boxtimes}\,\repX^-_{p{-}r}$}}
      \ar[1,-1]
      \\
      *{}&{\repX^+_{r}\makebox[0pt][l]{${\boxtimes}\,\repX^+_{r}$}}
      &*{}&*{}&*{}
      *{}&{\repX^-_{p{-}r}\makebox[0pt][l]{${\boxtimes}\,\repX^-_{p{-}r}$}}
    }
  \end{equation*}
\end{small}%
With respect to the right action, the picture is totally symmetric,
but with the subquotients placed as above, the structure of their
extensions has to be drawn as

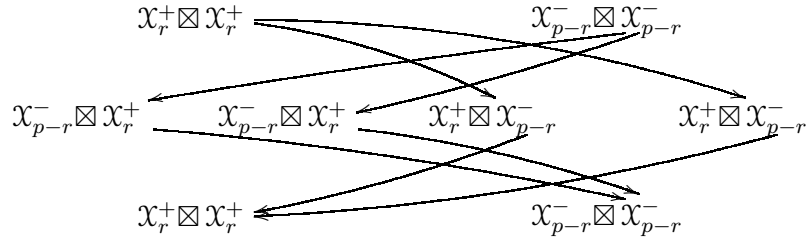
\begin{small}\vspace*{-.5\baselineskip}
  \begin{equation*}
    \mbox{}\kern50pt
    \xymatrix@=16pt{%
      *{}&{\makebox[0pt][r]{$\repX^+_{r}{\boxtimes}\,$}\repX^+_{r}}
      \ar@/^5pt/;[1,4]+<-12pt,8pt>
      \ar@/^8pt/;[1,6]+<-12pt,8pt>
      &*{}&*{}&*{}&*{}
      *{}&{\makebox[0pt][r]{$\repX^-_{p{-}r}{\boxtimes}\,$}\repX^-_{p{-}r}}
      \ar@/_1pt/-<10pt,5pt>;[1,-6]+<10pt,7pt>
      \ar@/^2pt/-<5pt,5pt>;[1,-4]
      \\
      {\makebox[0pt][r]{$\repX^-_{p{-}r}{\boxtimes}\,$}\repX^+_{r}}
      \ar@/^3pt/+<12pt,-4pt>;[1,6]+<-10pt,7pt>
      &*{}&
      {\makebox[0pt][r]{$\repX^-_{p{-}r}{\boxtimes}\,$}\repX^+_{r}}
      \ar@/^3pt/+<12pt,-4pt>;[1,4]+<-5pt,9pt>
      &*{}&*{}
      &{\makebox[0pt][r]{$\repX^+_{r}{\boxtimes}\,$}\repX^-_{p{-}r}}
      \ar@/^3pt/-<0pt,6pt>;[1,-4]
      &*{}
      &{\makebox[0pt][r]{$\repX^+_{r}{\boxtimes}\,$}\repX^-_{p{-}r}}
      \ar@/^6pt/-<0pt,6pt>;[1,-6]
      \\
      *{}&{\makebox[0pt][r]{$\repX^+_{r}{\boxtimes}\,$}\repX^+_{r}}
      &*{}&*{}&*{}&*{}
      *{}&{\makebox[0pt][r]{$\repX^-_{p{-}r}{\boxtimes}\,$}\repX^-_{p{-}r}}
    }
  \end{equation*}%
\end{small}%

Pictures of this type immediately yield the number of central elements
\textit{and their asso\-ciative algebra structure}.  First, each block
yields a primitive idempotent $\idem_I$, which is just the projector
on this block; second, there are maps sending $A\boxtimes B$ bimodules
into ``the same'' bimodules at lower levels, yielding nilpotent
central elements.  For $\UresSL2$, the bimodule decomposition contains
$p-1$ blocks of the above structure, plus two more given by
$\repX^+_p\boxtimes\repX^+_p$ and $\repX^-_p\boxtimes\repX^-_p$; in
each of the ``complicated'' blocks, there are two bimodule
automorphisms under which either the top $\repX^+_r\boxtimes\repX^+_r$
or the top $\repX^-_{p-r}\boxtimes\repX^-_{p-r}$ goes into the
corresponding bottom one, yielding two two central elements
$\nilp^+_r$ and $\nilp^-_{r}$ with zero products among themselves.
Therefore, the $(3p-1)$-dimensional center decomposes into a direct
sum of associative algebras as
\begin{equation*}
  \cZ_{\UresSL2}=\algI_{p}^{(1)}\oplus
  \algI_{0}^{(1)}\oplus\bigoplus_{r=1}^{p-1}\algB_{r}^{(3)},
\end{equation*}
where the dimension of each algebra is shown as a superscript.

\begin{otherqg}{}
  For $\XXX$, there are several intermediate levels in the filtration
  of projective modules, and hence the nilpotent elements are more
  numerous and have a nontrivial multiplication table
  (see~\cite{[FGST-q]} for the details); the center is
  $\half(3\pplus\,{-}\,1)(3\pminus\,{-}\,1)$-dimensional and
  decomposes into a direct sum of associative algebras as
  \begin{equation}\label{center-decomposition}
    \cZ_{\XXX}=\algI_{\pplus,\pminus}^{(1)}\oplus
    \algI_{0,\pminus}^{(1)}\oplus
    \bigoplus_{r=1}^{\pplus-1}\algB_{r,\pminus}^{(3)}
    \oplus
    \bigoplus_{r'=1}^{\pminus-1}\algB_{\pplus,r'}^{(3)}
    \oplus
    \bigoplus_{r,r'\in\setR}\!\algA_{r,r'}^{(9)},
  \end{equation}
  where the dimension of each algebra is shown as a superscript (and
  where
  $\vert\setR\vert=\half(\pplus\,{-}\,1)(\pminus\,{-}\,1)$).\linebreak[3]
  The $\algB^{(3)}$ algebras are just as in the previous formula, and
  each $\algA_{r,r'}^{(9)}$ is spanned by a primitive idempotent
  $\Idem_{r,r'}$ \textup{(}acting as identity on
  $\algA_{r,r'}^{(9)}$\textup{)} and eight radical elements
  $\vNE_{r,r'}$, $\vSW_{r,r'}$, $\vNW_{r,r'}$, $\vSE_{r,r'}$,
  $\wUp_{r,r'}$, $\wRight_{r,r'}$, $\wDown_{r,r'}$, $\wLeft_{r,r'}$
  that have the nonzero products
  \begin{alignat*}{2}
    \vNE_{r,r'}\vNW_{r,r'}&=\wUp_{r,r'},&
    \qquad
    \vNE_{r,r'}\vSE_{r,r'}&=\wRight_{r,r'},
    \\
    \vSW_{r,r'}\vNW_{r,r'}&=\wLeft_{r,r'},&
    \qquad
    \vSW_{r,r'}\vSE_{r,r'}&=\wDown_{r,r'}.
  \end{alignat*}%
\end{otherqg}

\section{Drinfeld map and factorizable and ribbon structures}
\label{sec:drinfeld}
\subsection{$M$-matrix and the Drinfeld map}\label{sec:M} For a
quasitriangular Hopf algebra $U$ with the universal $R$-matrix $R$,
the $M$-matrix is the ``square'' of the $R$-matrix, defined as
\begin{equation*}
  M=R_{21}R_{12}\in U\tensor U.
\end{equation*}
It satisfies the relations
\begin{equation}\label{M-commutes}
  \begin{split}
    (\Delta\tensor\id)(M)&=R_{32}M_{13}R_{23},\\
    M\Delta(x)&=\Delta(x) M\quad \forall x\in U.
  \end{split}
\end{equation}%
\begin{small}%
  Indeed, using \eqref{eq:quasi-2-d}, we find $(\Delta\otimes\id)(
  R_{21})= R_{32} R_{31}$ and then using~\eqref{eq:quasi-1} we obtain
  the first relation in~\eqref{M-commutes}.  Next, it follows
  from~\eqref{eq:def-prop-R-d} that $R_{21}R_{12}\Delta(x)
  =(R_{12}\Delta(x))^{\mathrm{op}}R_{12}
  =(\Delta^{\mathrm{op}}(x)R_{12})^{\mathrm{op}}R_{12}
  =\Delta(x)R_{21}R_{12}$, that is, the second relation
  in~\eqref{M-commutes}%
\end{small}.

The Drinfeld map $\Drinfeld:U^*\to U$ is defined as
\begin{equation*}
  \Drinfeld:\beta\mapsto(\beta\tensor\id)(M),
\end{equation*}
that is, if we write the $M$-matrix as
\begin{equation}\label{M-factorizable}
  M=\sum_I \pbwd_I\tensor \pbwdd_I,
\end{equation}
then $\Drinfeld(\beta)=\sum_I \beta(\pbwd_I)\pbwdd_I$.

Whenever $\Drinfeld:U^*\to U$ is an isomorphism of vector spaces, the
Hopf algebra~$U$ is called \textit{factorizable}~\cite{[RSts]}.
Equivalently, this means that $\pbwd_I$ and~$\pbwdd_I$
in~\eqref{M-factorizable} are two \textit{bases} in~$U$.

\begin{lemma}[\cite{[Drinfeld]}]
  In a factorizable Hopf algebra $U$, by restriction to $\Ch$
  \textup{(}see~\bref{qchar-def}\textup{)}, the Drinfeld map defines a
  homomorphism
  \begin{equation*}
    \Ch(U)\rightarrow\cZ(U)
  \end{equation*}
  of associative algebras.
\end{lemma}
\begin{proof}\small
  We first show that $\Drinfeld(\beta)$ is central for any
  $\beta\in\Ch$: for any $x\in U$, we calculate $\Drinfeld(\beta)x
  =\sum_I\beta(\pbwd_I)\pbwdd_I x =\sum_I\beta(\pbwd_I
  x''S^{-1}(x'))\pbwdd_I x'''$.  But because $M\Delta(x)=\Delta(x)M$
  and $\beta(xy)=\beta(S^2(y)x)$, we obtain that $\Drinfeld(\beta)x
  =\sum_I\beta(S(x')x''\pbwd_I)x'''\pbwdd_I=x\Drinfeld(\beta)$.

  Next, to show that $\Drinfeld:\Ch\to\cZ$ is a homomorphism of
  associative algebras, we recall that the product of two functionals
  is defined as $\beta\gamma(x)=(\beta\tensor\gamma)(\Delta(x))$, and
  therefore, using the first relation in~\eqref{M-commutes}, we have
  $\Drinfeld(\beta\gamma)
  =(\beta\tensor\gamma\tensor\id)((\Delta\tensor\id)(M))
  =(\beta\tensor\gamma\tensor\id)(R_{32}M_{13}R_{23})
  =(\gamma\tensor\id)(R_{21}\Drinfeld(\beta)R_{12})
  =\Drinfeld(\beta)(\gamma\tensor\id)(R_{21}R_{12})
  =\Drinfeld(\beta)(\gamma\tensor\id)(M)
  =\Drinfeld(\beta)\Drinfeld(\gamma)$.\footnote{It was noted
    in~\cite{[Schn]} that ``Drinfeld's proof of~\cite[3.3]{[Drinfeld]}
    shows more than what is actually stated
    in~\cite[3.3]{[Drinfeld]}.''  It actually follows that
    $\Drinfeld(\beta\gamma)=\Drinfeld(\beta)\Drinfeld(\gamma)$
    whenever $\beta\in\Ch(U)$ and $\gamma\in U^*$.}
\end{proof}

For the quantum groups considered here, the above homomorphism is in
fact an \textit{isomorphism} (cf.~\cite{[Drinfeld],[Schn]}).

\subsection{Kazhdan--Lusztig-dual quantum groups: Drinfeld's double,
  $M$-matrix, and $R$-matrix}\label{sec:R2M}
The quantum groups $U$ originating from logarithmic conformal models
\textit{are not quasitriangular}, but are nevertheless factorizable in
the following sense: the $M$-matrix can be expressed through an $R$
that is the universal $R$-matrix of a somewhat larger quantum
group~$\bar D$.\footnote{The standard definition of a factorizable
  quantum group~\cite{[RSts]} involves the universal $R$-matrix as
  well, which is the reason why we express some caution; the
  $R$-matrix in the $M$-matrix property~\eqref{M-commutes} is
  \textit{not} an element of $U\tensor U$.  In particular, $U$ is not
  unimodular in our case (but $U^*$ is!).}  This is true for both
$\UresSL2$ and $\XXX$, with the extension to $\bar D$ realized in each
case by introducing the generator $k=K^{1/2}$.  In other words, in
each case, there is a \textit{quasitriangular} quantum group $\bar D$
with a set of generators $k,\dots$, with a universal $R$-matrix $R$,
such that $R_{21}R_{12}$ turns out to belong to $U\tensor U$, where
$U$ is the Hopf subalgebra in $\bar D$ generated by $K=k^2$ and the
other $\bar D$ generators.  In the respective cases, $U$ is either
$\UresSL2$ or $\XXX$.

The universal $R$-matrix for $\bar D$, in turn, comes from
constructing the Drinfeld double~\cite{[Dr-87]} of the quantum group
$B$ generated by screenings in the logarithmic conformal field
model~\cite{[FGST],[FGST3]}.\footnote{The screenings generate only the
  upper-triangular subalgebra of the Kazhdan--Lusztig-dual quantum
  group; to these upper-triangular subalgebra, we add Cartan
  generator(s) constructed from zero modes of the free fields
  involved in the chosen free-field realization.  This gives the $B$
  quantum group.}  For $(p,1)$ models, $B$ is the Taft Hopf algebra
with generators $E$ and~$k$, with $kEk^{-1}=\q E$, $E^p=0$, and
$k^{4p}=\one$.  We then take the dual space $B^*$, which is a Hopf
algebra with the multiplication, comultiplication, unit, counit, and
antipode given~by
\begin{equation}\label{double-def}
  \begin{gathered}
    \coup{\beta\gamma}{x}=\coup{\beta}{x'}
    \coup{\gamma}{x''},\quad
    \coup{\Delta(\beta)}{x\tensor y}=\coup{\beta}{yx},\\
    \coup{1}{x}=\epsilon(x),\quad \epsilon(\beta)=\coup{\beta}{1},\quad
    \coup{S(\beta)}{x}=\coup{\beta}{S^{-1}(x)}
  \end{gathered}
\end{equation}
for $\beta,\gamma\in B^*$ and $x,y\in B$.  The Drinfeld double $D(B)$
is a Hopf algebra with the underlying vector space $B^*\tensor B$ and
with the multiplication, comultiplication, unit, counit, and antipode
given by those in $B$, by Eqs.~\eqref{double-def}, and~by
\begin{equation}\label{double-def-1}
  x\beta=\beta(S^{-1}(x''')?x')x'',
  \qquad x\in B,\quad\beta\in B^*.
\end{equation}
The resulting Hopf algebra $D(B)$ is canonically endowed with the
universal $R$-matrix~\cite{[Dr-87]}.

The doubling procedure also introduces the dual $\varkappa$ to the
Cartan element $k$, which is then to be eliminated by passing to the
quotient over (the Hopf ideal generated by) $k\kappa-\one$ (it follows
that $k\kappa$ is central in the double).  The quotient $\bar D$ is
still quasitriangular, but evaluating the $M$-matrix and the ribbon
element for it shows that they are turn out to be those for the (Hopf)
subalgebra generated by $K\equiv k^2$ and the other $\bar D$
generators, which is finally the Kazhdan--Lusztig-dual quantum group.
This was how the Kazhdan--Lusztig-dual quantum groups, together with
the crucial structures on them, were derived in~\cite{[FGST],[FGST3]}.

For $\UresSL2$, for example, the $M$-matrix is explicitly expressed in
terms of the PBW basis as
\begin{multline}\label{bar-M}
  M =\ffrac{1}{2p}
  \sum_{m=0}^{p-1}\sum_{n=0}^{p-1}
  \sum_{i=0}^{2p-1}\sum_{j=0}^{2p-1}
  \ffrac{(\q - \q^{-1})^{m + n}\!\!}{[m]! [n]!}\,
  \q^{m(m - 1)/2 + n(n - 1)/2}\\*
  \times \q^{- m^2 - m j + 2n j - 2n i - i j + m i} 
  F^{m} E^{n} K^{j}\tensor E^{m} F^{n} K^{i}.
\end{multline}

\subsection{Drinfeld-map images of traces and pseudotraces}
In a factorizable Hopf algebra, it follows that the Drinfeld-map
images of the traces over irreducible representations form an algebra
isomorphic to the Grothendieck ring.  Thus, there are central elements
\begin{alignat*}{2}
  \Drinfeld^\pm_r&=\Drinfeld(\Tr_{\repX^{\pm}_r}(\balance^{-1}?)),
  &\quad &1\leq r\leq p\\[-4pt]
  \intertext{for $\UresSL2$ and }
  \Drinfeld^\pm_{r,r'}&=\Drinfeld(\Tr_{\repX^{\pm}_{r,r'}}(\balance^{-1}?)),
  &\quad &1\leq r\leq p, \ 1\leq r'\leq p'
\end{alignat*}
for $\XXX$, which satisfy the respective algebra~\eqref{the-fusion}
and~\eqref{the-Groth}.

\subsubsection{}In $\UresSL2$, for example, Eq.~\eqref{bar-M} allows
us to calculate the $\Drinfeld^{\alpha}_{s}$ explicitly,
\begin{multline}\label{the-cchi}
  \Drinfeld^{\alpha}_{s} = \alpha^{p+1}(-1)^{s+1}\sum_{n=0}^{s-1}
  \sum_{m=0}^{n}
  (\q-\q^{-1})^{2m} \q^{-(m+1)(m+s-1-2n)}\times{}\\*
  {}\times\qbin{s\!-\!n\!+\!m\!-\!1}{m}
  \qbin{n}{m} E^m F^m K^{s-1+\beta p - 2n + m}
\end{multline}
(where $\beta=0$ if $\alpha=+1$ and $\beta=1$ if $\alpha=-1$).  In
particular, $\Drinfeld^{+}_{2}=-\hcas$, where $\hcas$ is the Casimir
element, Eq.~\eqref{hCas}.  The fact that the $\Drinfeld^{\pm}_{r}$
given by~\eqref{the-cchi} satisfy Grothendieck-ring
relations~\eqref{the-fusion} implies a certain $q$-binomial identity,
see~\cite{[FGST]}.

\begin{rem}
  The $\UresSL2$ Casimir element satisfies the minimal polynomial
  relation $\Psi_{2p}(\hcas)=0$, with $\Psi_{2p}$ in~\eqref{Psi}.
  This relation, with $p-1$ multiplicity-$2$ roots of $\Psi_{2p}$,
  allows constructing a basis in the center $\cZ$ of $\UresSL2$
  consisting of primitive idempotents $\idem_r$ and elements $\nilp_r$
  in the radical of the associative commutative
  algebra~$\cZ$~\cite{[FGST]} (see~\cite{[Kerler]} and
  also~\cite[Ch.~V.2]{[Gant]}).  For this, we define the polynomials
  \begin{gather*}
    \psi_0(x)=(x-\hbeta_p)\prod_{r=1}^{p-1}(x-\hbeta_r)^2,
    \qquad
    \psi_p(x)=(x-\hbeta_0)\prod_{r=1}^{p-1}(x-\hbeta_r)^2,
    \\[-4pt]
    \psi_s(x)=(x-\hbeta_0)\,(x-\hbeta_p)
    \prod_{\substack{r=1\\
        r\neq s}}^{p-1}(x-\hbeta_r)^2,\quad 1\leq s\leq p-1,    
  \end{gather*}
  where we recall that all $\hbeta_j$ are distinct.  Then the
  canonical elements in the radical of~$\cZ$ are
  \begin{align*}
    \nilp^{\pm}_s&=\pi^{\pm}_s \nilp_s,\qquad
    1\leq s\leq p-1,\\[-4pt]
  \intertext{where}
  \nilp_s &= \smash[t]{\ffrac{1}{\psi_s(\hbeta_s)}
    \bigl(\hcas-\hbeta_s\bigr)\psi_s(\hcas)}
  \end{align*}
  and we introduce the projectors
  \begin{equation*}
    \pi^+_s=\ffrac{1}{2p}\sum_{n=0}^{s-1}
    \sum_{j=0}^{2p-1}\q^{(2n-s+1)j}K^j
    ,\qquad
    \pi^-_s= \ffrac{1}{2p}\sum_{n=s}^{p-1}
    \sum_{j=0}^{2p-1}\q^{(2n-s+1)j}K^j,
  \end{equation*}
  and the canonical central idempotents are given by
  \begin{equation*}
    \idem_s=\ffrac{1}{\psi_s(\hbeta_s)}
    \bigl(\psi_s(\hcas) - \psi'_s(\hbeta_s)\nilp_s\bigr),
    \quad 0\leq s\leq p,
  \end{equation*}
  where we formally set $\nilp_0=\nilp_p=0$.  \begin{otherqg}{}A
    similar construction exists for the center of
    $\XXX$~\cite{[FGST-q]}, where, in particular, there are not two
    but four types of projectors $\piUp_{r, r'}$, $\piLeft_{r, r'}$,
    $\piRight_{r, r'}$, and $\piDown_{r, r'}$; for either algebra,
    these are projectors on the weights occurring in irreducible
    modules in the full subcategory labeled by the subscript%
  \end{otherqg}.
\end{rem}

\subsubsection{``Drinfeld'' basis}Applied to the basis $\gamma_A$ of
traces and pseudotraces in $\Ch$, the Drinfeld map gives a basis
\begin{equation*}
  \Drinfeld_A=\Drinfeld(\gamma_A)
\end{equation*}
in the center~$\cZ$.

This ``Drinfeld'' basis (which is not defined uniquely because
pseudotraces are not defined uniquely) specifies an explicit splitting
of the associative commutative algebra $\cZ$ into the Grothendieck
ring and its linear complement.  The products of the Grothendieck ring
elements with elements from the complement may also be of significance
in the Kazhdan--Lusztig context.  The full algebra of $q$-characters
(traces and pseudotraces) for $\UresSL2$, mapped into the center by
the Drinfeld map, is evaluated in~\cite{[GT]}; it can be understood as
a generalized fusion, to be compared with a recent calculation to this
effect in the logarithmic $(p,1)$ models in~\cite{[FK]}.

Under $\modS\in\SLiiZ$ acting as in~\eqref{3-diagram}, clearly, the
Drinfeld basis elements are mapped into the Radford basis,
\begin{equation*}
  \modS:\Drinfeld_A\mapsto\Radford_A.
\end{equation*}
Realizing $\modT\in\SLiiZ$ on the center requires yet another
structure, the ribbon element.

\subsection{Ribbon structure}
A \textit{ribbon element} \cite{[RT]} is a $\ribbon\in\cZ$ such that
\begin{equation*}
  \Delta(\ribbon)=M^{-1}(\ribbon\tensor\ribbon),
\end{equation*}
with $\counit(\ribbon)=1$ and $S(\ribbon)=\ribbon$ (and $\ribbon^2=
\sqs S(\sqs)$, see~\bref{app:S2}).  The procedure for finding the
ribbon element involves two steps: we first find the canonical
element~\eqref{canon-sqs} (which involves the universal $R$-matrix for
the larger, quasitriangular quantum group~$\bar D$ mentioned
in~\bref{sec:R2M}) and then evaluate the balancing element $\balance$
(see~\bref{sec:balancing}) in accordance with Drinfeld's
Lemma~\eqref{Dr-lemma}, from the comodulus obtained from the explicit
expression for the integral (this is the job done by the comodulus).
Then
\begin{equation}\label{balance-ribbon}
  \ribbon=\sqs\balance^{-1}.
\end{equation}
It follows, again, that $\ribbon$ is an element of a Hopf subalgebra
in $\bar D$, which is $\UresSL2$ or~$\XXX$.

\subsubsection{}For $\UresSL2$, where $\balance=K^{p+1}$, we
have~\cite{[FGST]}
\begin{equation}\label{sl2-ribbon}
  \ribbon
  =\sum_{s=0}^{p}(-1)^{s+1} \q^{-\half(s^2-1)}\idem_s
  +\sum_{s=1}^{p - 1} (-1)^p \q^{-\half(s^2 - 1)}
  [s]\,\ffrac{\q - \q^{-1}}{\sqrt{2 p}}\,\vvarphi_{s},
\end{equation}
where $\idem_s$ are the canonical idempotents in the center and 
\begin{equation}\label{vvraphi}
  \vvarphi_{s} =
  \ffrac{p-s}{p}\,\Radford^{\,+}_{s} - \ffrac{s}{p}\,\Radford^{\,-}_{p-s},
  \quad 1\leq s\leq p-1,
\end{equation}
are nilpotent central elements expressed through the Radford-map
images $\Radford^{\,\pm}_{s}$ of the (traces over) irreducible
representations $\repX^{\pm}_s$.

\begin{rem}
  The above form of $\ribbon$ implies that~\cite{[FGST2]}
  \begin{equation*}
    \ribbon=e^{2i\pi L_0}
  \end{equation*}
  (where $L_0$ is the zero-mode Virasoro generator in the $(p,1)$
  logarithmic conformal model); in particular, the exponents involving
  $s^2$ in~\eqref{sl2-ribbon} are simply related to conformal
  dimensions of primary fields.  Rather interestingly, the
  nonsemisimple action of $L_0$ on the lattice vertex operator algebra
  underlying the construction of the logarithmic $(p,1)$ model is thus
  correlated with the decomposition of the ribbon element with respect
  to the central idempotents and nilpotents.
\end{rem}

\subsubsection{}
\begin{otherqg}{}
  For $\XXX$, the ribbon element is given by
  \begin{equation*}
    \ribbon=
    \sum_{(r,r')\in\setI}\!\!\!
    e^{2i\pi\Delta_{r,r'}}\Idem_{r,r'} +\text{nilpotent terms}, 
  \end{equation*}
  where $\Idem_{r,r'}$ are the $\half(\pplus \,{+}\,1)(\pminus \,{+}\,1)$
  primitive idempotents in the associative commutative algebra $\cZ$
  \textup{(}the explicit form of the nilpotent terms being not very
  illuminating at this level of detail, see~\cite{[FGST-q]} for the
  full formula\textup{)}, and
  \begin{equation*}
    \Delta_{r,r'}
    =\ffrac{(\pplus r'\,{-}\,\pminus r)^2-(\pplus\,{-}\,\pminus)^2}{
      4\pplus\pminus} 
  \end{equation*}
  are conformal dimensions of primary fields borrowed from the
  logarithmic model~\cite{[FGST3]}.
\end{otherqg}

\section{Modular group action}\label{sec:modular}
\subsection{Defining the action} In defining the modular group action
on the center we follow~\cite{[Lyu],[LM],[Kerler]} with an
insignificant variation in the definition of $\modT$, introduced
in~\cite{[FGST],[FGST3]} in order to simplify comparison with the
modular group representation generated from characters of the chiral
algebra in the corresponding logarithmic conformal model.  On the
quantum group center, the $\SLiiZ$-action is defined by
\begin{equation}\label{theST}
  \begin{split}
    \modS:{}& x\mapsto \Radford\bigl(\Drinfeld^{-1}(x)\bigr),\\
    \modT:{}& x\mapsto\AT\,\modS(\ribbon\modS^{-1}(x)),
  \end{split}
\end{equation}
where $c$ is the central charge of the conformal model,
\begin{otherqg}{}
  e.g.,
  \begin{equation*}
    c=13-6\ffrac{\pplus }{\pminus } - 6\ffrac{\pminus }{\pplus }
  \end{equation*}
  for the $(\pplus,\pminus)$ model%
\end{otherqg}.\footnote{Reversing the argument, for a factorizable
  ribbon quantum group that can be expected to correspond to a
  conformal field model, the normalization of $\modT$ (i.e., the
  factor accompanying the ribbon element) may thus indicate the
  central charge, and the decomposition of the ribbon element into the
  basis of primitive idempotents and elements in the radical is
  suggestive about the conformal dimensions.}

\subsection{Calculation results} The result of
evaluating~\eqref{theST} in each case gives the structure of the
$\SLiiZ$ representation of the type that was first noted
in~\cite{[Kerler]} for the small quantum $s\ell(2)$.\footnote{The
  small quantum groups have been the subject of some constant
  interest, see, e.g.,~\cite{[L-center],[BL],[O]} and the references
  therein.}

\subsubsection{}On the center of $\UresSL2$, the $\SLiiZ$
representation is given by~\cite{[FGST]}
\begin{equation}\label{SL2Zrep-p1}
  \cZ_{\UresSL2} = R_{p+1}\oplus\oC^2\tensor R_{p-1},
\end{equation}
where $\oC^2$ is the defining two-dimensional representation,
$\mathcal{R}_{p-1}$ is a $(p\,{-}\,1)$-dimen\-sional
$\SLiiZ$-representation (the ``$\sin\!\frac{\pi r s}{p}$''
representation, in fact, the one on the unitary $\hSL2_{k}$-characters
at the level $k\,{=}\,p\,{-}\,2$), and $\mathcal{R}_{p+1}$ is a
``$\cos\!\frac{\pi r s}{p}$'' \ $(p\,{+}\,1)$-dimensional
representation.

\begin{otherqg}{}
  On the center of $\XXX$, the $\SLiiZ$-representation structure is
  given by~\cite{[FGST-q]}
  \begin{equation}\label{SL2Zrep-pq}
    \cZ_{\XXX}
    =\Rmin\oplus\Rproj\oplus\oC^2\tensor(\RLeft\oplus\RRight)\oplus
    \oC^3\tensor\Rmin,
  \end{equation}
  where $\oC^3$ is the symmetrized square of $\oC^2$, $\Rmin$ is the
  $\frac{1}{2}(\pplus \,{-}\,1)(\pminus\,{-}\,1)$-dimensional
  $\SLiiZ$-representation on the characters of the \textit{rational}
  $(\pplus,\pminus )$ Virasoro model, and $\Rproj$, $\RLeft$, and
  $\RRight$ are certain $\SLiiZ$ representations of the respective
  dimensions $\frac{1}{2}(\pplus \,{+}\,1)(\pminus \,{+}\,1)$,
  $\frac{1}{2}(\pplus\,{+}\,1)(\pminus\,{-}\,1)$, and
  $\frac{1}{2}(\pplus\,{-}\,1)(\pminus\,{+}\,1)$.
\end{otherqg}

As noted above, \eqref{SL2Zrep-p1} and~\eqref{SL2Zrep-pq} coincide
with the respective $\SLiiZ$-representations on generalized characters
of $(p,1)$ and $(\pplus,\pminus)$ logarithmic conformal field models
evaluated in~\cite{[FGST],[FGST3]}.

\subsubsection{}
The role of the subrepresentations identified in~\eqref{SL2Zrep-p1}
and~\eqref{SL2Zrep-pq} is yet to be understood from the quantum-group
standpoint, but it is truly remarkable in the context of the
Kazhdan--Lusztig correspondence.  The occurrence of the $\oC^n$ tensor
factors is rigorously correlated with the fact that the
$\psi_{b'}(\tau)$ functions in~\eqref{modS-char1}--\eqref{modS-char2}
are given by (certain linear combinations of) characters \textit{times
  polynomials in~$\tau$ of degree~$n\,{-}\,1$.}

In the quantum group center, the subrepresentations
in~\eqref{SL2Zrep-p1} and~\eqref{SL2Zrep-pq} are described as the span
of certain combinations of the elements of ``Radford'' and
``Drinfeld'' bases $\Radford_A$ and $\Drinfeld_A$.  For $\UresSL2$, in
particular, the central elements~\eqref{vvraphi}, together with their
$\modS$-images $\frac{p-s}{p}\,\Drinfeld^{\,+}_{s} -
\frac{s}{p}\,\Drinfeld^{\,-}_{p-s}$, $1\leq s\leq p\,{-}\,1$, span the
$\oC^2\tensor R_{p-1}$ representation; in the logarithmic $(p,1)$
model, the same representation is realized on the $2(p\,{-}\,1)$
functions
\begin{equation*}
  \tau\bigl(\ffrac{p-s}{p}\,\chi^{\,+}_{s}(\tau) -
  \ffrac{s}{p}\,\chi^{\,-}_{p-s}(\tau)\bigr),\qquad
  \ffrac{p-s}{p}\,\chi^{\,+}_{s}(\tau) -
  \ffrac{s}{p}\,\chi^{\,-}_{p-s}(\tau),
\end{equation*}
where $\chi^\pm_r(\tau)$ are the triplet algebra
characters~\cite{[FGST]}.  On the other hand, the
$(p\,{+}\,1)$-dimen\-sion\-al representation $R_{p+1}$ in the center
is linearly spanned by $\Drinfeld^{\pm}_{p}$ and $\Drinfeld^{+}_{s} +
\Drinfeld^{-}_{p-s}$, $1\leq s\leq p\,{-}\,1$ (the ideal already
mentioned after~\eqref{basis-P}); in the $(p,1)$ model, the same
representation is realized on the linear combinations of characters
\begin{equation*}
  \chi^{\pm}_{p}(\tau),\qquad
  \chi^{+}_{s}(\tau) + \chi^{-}_{p-s}(\tau).
\end{equation*}

\begin{otherqg}{}
  The $\XXX$ setting in~\cite[Sec.~\textbf{5.3}]{[FGST-q]} gives
  rather an abundant picture of how the various traces \textit{and
    pseudotraces}, mapped into the center, combine to produce the
  subrepresentations and how precisely these linear combinations
  correspond to the characters \textit{and extended characters} in the
  logarithmic $(\pplus,\pminus)$ model.\footnote{Once again:
    $\oC$-linear combinations of the
    $\half(3\pplus\,{-}\,1)(3\pminus\,{-}\,1)$ traces and pseudotraces
    (mapped to the center) carry the same $\SLiiZ$-representations as
    certain $\oC[\tau]$-linear combinations of the $2\pplus\pminus +
    \half(\pplus\,{-}\,1)(\pminus\,{-}\,1)$ \textit{characters} of the
    $W$-algebra; the total dimension is
    $\half(3\pplus\,{-}\,1)(3\pminus\,{-}\,1)$ in either case.}  Here,
  we only note the $\Rproj$ representation, of dimension $\half(\pplus
  \,{+}\,1)(\pminus\,{+}\,1)$, linearly spanned by
  $\Drinfeld^{+}_{r,r'} + \Drinfeld^{-}_{\pplus-r,r'} +
  \Drinfeld^{-}_{r,\pminus-r'} + \Drinfeld^{+}_{\pplus-r,\pminus-r'}$
  (with $(r,r')\in\setR$, where
  $\vert\setR\vert=\half(\pplus-1)(\pminus-1)$),
  $\Drinfeld^{+}_{r,\pminus} + \Drinfeld^{-}_{\pplus - r,\pminus}$
  (with $1\leq r\leq \pplus-1$), $\Drinfeld^{+}_{\pplus,r'} +
  \Drinfeld^{-}_{\pplus,\pminus-r'}$ ($1\leq r'\leq \pminus\,{-}\,1$),
  and $\Drinfeld^{\pm}_{\pplus,\pminus}$.  In the logarithmic
  $(\pplus,\pminus)$ model, the same $\SLiiZ$-representation is
  realized on the linear combinations of $W$-algebra characters
  \begin{alignat*}{2}
    &\chi_{r,r'}(\tau)+ 2\chi^+_{r,r'}(\tau)
    +2\chi^-_{r,\pminus-r'}(\tau) +2\chi^-_{\pplus -r,r'}(\tau)
    +2\chi^+_{\pplus -r,\pminus-r'}(\tau),&\quad&(r,r')\in\setii,
    \\
    & 2\chi^+_{\pplus, \pminus-r'}(\tau)+2\chi^-_{\pplus, r'}(\tau),
    &&1\leq r'\leq \pminus\,{-}\,1,
    \\
    & 2\chi^+_{\pplus -r,\pminus}(\tau)
    +2\chi^-_{r,\pminus}(\tau),&&1\leq r\leq \pplus \,{-}\,1,
    \\
    &2\chi^{\pm}_{\pplus, \pminus}(\tau)&
  \end{alignat*}
  (with the same size
  $|\setii|=\half(\pplus\,{-}\,1)(\pminus\,{-}\,1)$ of the index set),
  where $\chi_{r,r'}(\tau)$ are the characters of the Virasoro
  \textit{rational} model and $\chi^{\pm}_{r,r'}(\tau)$ are the other
  $2\pplus\pminus$ characters of the $W$-algebra~\cite{[FGST3]}.  The
  above combinations do not involve generalized characters (which
  occur where the $\oC^n$ factors are involved in the
  $\SLiiZ$-representation isomorphic to the one in~\eqref{SL2Zrep-pq}
  and which are in fact the origin of these $\oC^n$ factors from the
  conformal field theory standpoint).

  A remarkable feature of the $\SLiiZ$ representation on the $\XXX$
  center is the occurrence of $\Rmin$, the $\SLiiZ$-representation on
  the characters of the rational Virasoro model, even though the
  $\XXX$-representations $\repX^{\pm}_{r,r'}$ are in a $1:1$
  correspondence not with all the primary fields of the $W$-algebra in
  the logarithmic model but just with those \textit{except} the
  rational-model ones.
\end{otherqg}

\subsection{Beyond the quantum group}
Two algebraic structures on the quantum group center are most
important from the standpoint of the Kazhdan--Lusztig correspondence:
the modular group action and the Grothendieck ring (the latter is a
subring in the center spanned by Drinfeld-map images of the
irreducible representations).  The resulting Groth\-endieck rings, or
Verlinde algebras are nonsemisimple.

A classification of Verlinde algebras has been proposed in a totally
different approach, that of double affine Hecke algebras (Cherednik
algebras)~\cite{[Cherednik]}, where Verlinde algebras occur as certain
representations of Cherednik algebras; an important point is that a
modular group action is built into the structure of Cherednik
algebras.  It can thus be expected that the $(p,1)$-model fusion (the
$\UresSL2$ Grothendieck ring)~\eqref{the-fusion}, of dimension $2p$,
admits a realization associated with a Cherednik algebra
representation.  But because an isomorphic image of the Grothendieck
ring is contained in the center, a natural further question is whether
the entire $\UresSL2$ center, of dimension $3p\,{-}\,1$, endowed with
the $\SLiiZ$ action, is also related to Cherednik algebras.

It was shown in~\cite{[MT]} that the center $\cZ$ of $\UresSL2$, as an
associative commutative algebra and as an $\SLiiZ$ representation, is
indeed extracted from a representation space of the simplest Cherednik
algebra $\mathscr{H}$, defined by the relations
\begin{alignat*}{2}
  TXT&=X^{-1},&\quad TY^{-1}T&=Y,\\
  XY&=\q YX T^2,&\quad (T-\q)(T+\q^{-1})&=0
\end{alignat*}
on the generators $T$, $X$, $Y$, and their inverse.  In these terms,
the $PSL(2,\oZ)$ action is defined by the elements $\tau_+=\bigl(\!
\begin{smallmatrix}
  1&1\\
  0&1
\end{smallmatrix}\!\bigr)$ and $\tau_-=\bigl(\!
\begin{smallmatrix}
  1&0\\
  1&1
\end{smallmatrix}\!\bigr)$ being realized as the $\mathscr{H}$
automorphisms~\cite{[Cherednik]}
\begin{alignat*}{4}
  &\tau_+:& \ X&\mapsto X,&\quad Y&\mapsto \q^{-1/2} X Y,&
  \quad T&\mapsto T,\\
  &\tau_-:& \ X&\mapsto \q^{1/2}YX,&\quad Y&\mapsto Y,&
  \quad T&\mapsto T.
\end{alignat*}

For each $p\geq3$, the authors of~\cite{[MT]} construct a
$(6p-4)$-dimensional (reducible but indecomposable) representation of
$\mathcal{H}$ in which the eigensubspace of $T$ with eigenvalue $\q$
(as before, $\q=e^{i\pi/p}$) is $(3p-1)$-dimensional.  The associative
commutative algebra structure induced on this eigensubspace in
accordance with Cherednik's theory then coincides with the associative
commutative algebra structure on the center $\cZ$ of $\UresSL2$.
Furthermore, the $\SLiiZ$ representations constructed on this space
\`a la Cherednik and \`a la Lyubashenko coincide.  Also, the Radford-
and Drinfeld-map images of irreducible representations in the center
can be ``lifted'' to the level of~$\mathcal{H}$ (as eigenvectors of
$X+X^{-1}$ and $Y+Y^{-1}$ respectively)~\cite{[MT]}.

\section{Conclusions}
Without a doubt, it would be extremely useful to rederive the results
such as the equivalence of modular group representations in a more
``categorical'' approach; this would immediately suggest
generalizations.  But the quantum group ``next in the queue'' after
$\UresSL2$ and $\XXX$ is a quantum $s\ell(2|1)$ (cf.\ the remarks
in~\cite{[S-hsl2]}), which already requires extending many basic facts
(e.g., those in~\cite{[Lyu]}) to the case of quantum
\textit{super}groups.

The center of the Kazhdan--Lusztig-dual quantum group is to be
regarded as the center of the corresponding logarithmic conformal
field model; this calls for applications to boundary states in
logarithmic models.

\subsubsection*{Acknowledgments} I am grateful to M.~Finkelberg and
J.~Fuchs for useful comments and to V.~Schomerus for discussions.
Special thanks go to A.~Gainutdinov for his criticism.  This paper was
supported in part by the RFBR grant~07-01-00523.

\appendix
\section{}
\subsection{The center and $q$-characters}\label{qchar-def}
The center $\cZ$ of a Hopf algebra $U$ can be characterized as the set
\begin{equation}
  \cZ=\bigl\{ y\in U \bigm| \ad_x(y)
  =\epsilon(x)y\quad \forall x\in U\bigr\}.
\end{equation}
The space of $q$-characters $\Ch=\Ch(U)\subset U^*$ is defined as
\begin{multline}\label{Ch-def}
  \Ch=\bigl\{\beta\in U^* \bigm| \ad^*_x(\beta)
  =\counit(x)\beta\quad \forall x\in U\bigr\}\\
  = \bigl\{\beta\in U^* \bigm| \beta(xy)=\beta\bigl(S^2(y)x\bigr)
  \quad \forall x,y\in U\bigr\},
\end{multline}
where the coadjoint action $\ad^*_a:U^*\to U^*$ is
$\ad^*_a(\beta)=\beta\bigl(S(a')?a''\bigr)$, $a\in U$,
$\beta\in U^*$.

\subsection{Quasitriangular Hopf algebras}\label{app:quasi-tr}
Quasitriangular (or braided) Hopf algebras were introduced
in~\cite{[Dr-87]} (also see~\cite{[Rad-quasitr]}).  A quasitriangular
Hopf algebra $U$ has an invertible element $R\,{\in}\, U\tensor U$
satisfying
\begin{gather}
  \Delta^{\mathrm{op}}(x)=R\Delta(x) R^{-1},\label{eq:def-prop-R-d}
  \\
  (\Delta\otimes\id)(R)= R_{13} R_{23},\label{eq:quasi-1}
  \\
  (\id\otimes\Delta)(R)= R_{13} R_{12},\label{eq:quasi-2-d}
  \\
  \intertext{the Yang--Baxter equation}
  R_{12}R_{13}R_{23}=R_{23}R_{13}R_{12},
\end{gather}
and the relations
$(\epsilon\tensor\id)(R)=\one=(\id\tensor\epsilon)(R)$, $(S\tensor
S)(R)=R$.

\subsection{Square of the antipode~\cite{[Drinfeld]}}
\label{app:S2}
In any quasitriangular Hopf algebra, the square of the antipode is
represented by a similarity transformation
\begin{equation*}
  S^2(x)=  \sqs x\sqs^{-1},
\end{equation*}
where the \textit{canonical element} $\sqs$ is given by
\begin{gather}\label{canon-sqs}
  \sqs= \cdot\bigl((S\tensor\id)R_{21}\bigr),\quad
  \sqs^{-1}=\cdot\bigl((S^{-1}\tensor S)R_{21}\bigr)
\end{gather}
(where $\cdot(a\tensor b)=ab$) and satisfies the property
\begin{gather}\label{Delta-u}
  \Delta(\sqs)= M^{-1}(\sqs\tensor\sqs)=(\sqs\tensor\sqs)M^{-1}
\end{gather}
(where we recall that $M=R_{21}R_{12}$).


\subsection{Balancing element}\label{sec:balancing}
We also need the so-called \textit{balancing element} $\balance\in U$
that satisfies~\cite{[Drinfeld]}
\begin{equation}\label{balance-prop}
  \begin{gathered}
    S^2(x)=\balance x\balance^{-1}\quad\forall x\in U,\\
    \Delta(\balance)=\balance\tensor\balance,
  \end{gathered}
\end{equation}

The balancing element $\balance$ allows constructing the ``canonical''
$q$-character associated with any (irreducible, because traces are
insensitive to indecomposability) representation~$\repX$ as the
(``quantum'') trace
\begin{equation}\label{qCh}
  \Tr_{\repX}(\balance^{-1}?)
  \in\Ch(U).
\end{equation}

For a Hopf algebra $U$ with a right integral $\rint$, we recall the
definition of a comodulus in~\bref{sec:comodulus}.  Whenever a square
root of the comodulus $\comodul$ can be calculated in $U$, a lemma of
Drinfeld \cite{[Drinfeld]} states that
\begin{equation}\label{Dr-lemma}
  \balance^2=\comodul.
\end{equation}

\end{document}